\documentclass[a4paper,10pt,twocolumn,twoside]{IEEEtran}
\usepackage{cite}

\ifCLASSINFOpdf
\else
\fi
\usepackage[cmex10]{amsmath}
\usepackage{array}

\ifCLASSOPTIONcompsoc
  \usepackage[caption=false,font=normalsize,labelfont=sf,textfont=sf]{subfig}
\else
  \usepackage[caption=false,font=footnotesize]{subfig}
\fi
\usepackage{url}

\usepackage{tabulary}
\usepackage[english]{babel}
\usepackage{blindtext}
\usepackage[makeroom]{cancel}
\usepackage{amsfonts}
\usepackage{amsthm}  
\usepackage{amssymb}
\usepackage[pdftex]{graphicx}
\graphicspath{{../pdf/}{../jpeg/}}
\DeclareGraphicsExtensions{.pdf,.jpeg,.png}
\usepackage{pgf,tikz}
\usepackage{tabularx}
\usepackage{float}

\usepackage{enumitem} 

\newcommand{\vect}[1]{\boldsymbol{#1}}
\newcommand{\mat}[1]{\boldsymbol{#1}}
\newcommand{\wt}[1]{\widetilde{#1}}
\newcommand{\wh}[1]{\widehat{#1}}

\DeclareMathOperator{\sgn}{\text{sgn}}

\newtheorem{remark}{Remark}
\newtheorem{theorem}{Theorem}
\newtheorem{lemma}{Lemma}

\newtheorem{definition}{Definition}
\newtheorem{corollary}{Corollary}

\newtheorem{assumption}{Assumption}

\begin{document}

\title{Leader Tracking of Euler-Lagrange Agents on Directed Switching Networks Using A Model-Independent Algorithm}
%
%
%

\author{\IEEEauthorblockN{Mengbin Ye $\quad$ }
\and
\IEEEauthorblockN{Brian D.O. Anderson, \emph{Life Fellow, IEEE} $\quad$ }
\and
\IEEEauthorblockN{Changbin Yu, \emph{Senior Member, IEEE}}

\thanks{This work was supported by the Australian Research Council (ARC) under the ARC grants \mbox{DP-130103610} and \mbox{DP-160104500}, by the National Natural Science Foundation of China (grant 61375072), and by Data61-CSIRO (formerly NICTA). M. Ye is supported by an Australian Government Research Training Program (RTP) Scholarship.

M.Ye is with the Research School of Engineering, Australian National University. C.Yu and B.D.O.Anderson are with the Australian National University and with Hangzhou Dianzi University, Hangzhou, China. B.D.O. Anderson is also with Data61-CSIRO (formerly NICTA Ltd.), Canberra, Australia.}
\thanks{\texttt{\{Mengbin.Ye, Brian.Anderson, Brad.Yu\}@anu.edu.au}}
}

\maketitle

\begin{abstract}
In this paper, we propose a discontinuous distributed model-independent algorithm for a directed network of Euler-Lagrange agents to track the trajectory of a leader with non-constant velocity. We initially study a fixed network and show that the leader tracking objective is achieved semi-globally exponentially fast if the graph contains a directed spanning tree. By model-independent, we mean that each agent executes its algorithm with no knowledge of the parameter values of any agent's dynamics. Certain bounds on the agent dynamics (including any disturbances) and network topology information are used to design the control gain. This fact, combined with the algorithm's model-independence, results in robustness to disturbances and modelling uncertainties. Next, a continuous approximation of the algorithm is proposed, which achieves practical tracking with an adjustable tracking error. Last, we show that the algorithm is stable for networks that switch with an explicitly computable dwell time. Numerical simulations are given to show the algorithm's effectiveness.  
\end{abstract}

\begin{IEEEkeywords}
model-independent, euler-lagrange agent, directed graph, distributed algorithm, tracking, switching network
\end{IEEEkeywords}

%
\IEEEpeerreviewmaketitle

\section{Introduction}
%
%
%
%

\IEEEPARstart{C}{oordination} of multi-agent systems using distributed algorithms has been widely studied over the past decade \cite{ren2011distributed_book}. Of recent interest is the study of agents whose dynamics are described using Euler-Lagrange equations of motion, which from here onwards will be referred to as Euler-Lagrange agents (in some literature known as Lagrangian agents). The non-linear Euler-Lagrange equation can be used to model the dynamics of a large class of mechanical, electrical and electro-mechanical systems \cite{ortega2013passivity_book}. Thus, there is significant motivation to study coordination problems with multiple Euler-Lagrange agent. The interaction between agents may be modelled using a graph \cite{ren2011distributed_book}, and the agents collectively form a network. Directed networks capture unilateral interactions (e.g. sensing or communication) between agents and are generally more desirable when compared to undirected networks.


To better place our results in context, two existing approaches for designing coordination algorithms for Euler-Lagrange networks are reviewed: model-dependent and adaptive algorithms. The aim is to give readers an idea of available works; the list is not exhaustive. The papers \cite{chung2009cooperative,hu2015FT_EL_Tracking,yang2017EL_global_track} study different coordination objectives, such as consensus or leader tracking, using algorithms that require \emph{exact knowledge of the agent models}. Specifically, each agent's algorithm requires knowledge of its own Euler-Lagrange equation in order to execute. The algorithms are therefore less robust to uncertainties in the model, e.g. some parameters in the Euler-Lagrange equation may be unknown or uncertain. Recently, the more popular approach is for each agent to use an \emph{adaptive algorithm}. Specifically, an Euler-Lagrange equation can be linearly parametrised \cite{ortega2013passivity_book} with respect to a set of constant parameters of the equation, e.g. the mass of an arm on a robotic manipulator agent. This parametrisation is then used in an adaptive algorithm to allow the agent to estimate its own set of constant parameters  (which is assumed to be unknown) while simultaneously achieving the multi-agent coordination objective. Using adaptive algorithms, containment control was studied in \cite{meng2012leader,mei2012distributed}, while leaderless consensus was studied in \cite{mei2012distributed,abdessameud2014synchronization}. Leader tracking algorithms were studied in \cite{chen2011distributed,nuno2011synchronization,meng2014leader_el,ghapani2016el_flocking_leader,abdessameud2016TAC_eltrack}.

In contrast to the above works, which rely on direct knowledge (or adaptive identification) of an agent model, there have been relatively few works studying \emph{model-independent} algorithms, that is, algorithms for obtaining robust controllers. Furthermore, most results study model-independent algorithms on \emph{undirected} networks. The pioneering work in \cite{ren2009distributed} considered leaderless position consensus, with time-delay considered in \cite{nuno2013consensus_el}. Consensus to the intersection of target sets is studied in \cite{meng2014set_target_aggregation_EL}. 
Leader-tracking algorithms are studied in \cite{mei2011distributed_EL,zhao2015FT_ELTracking,klotz2015_leaderfollower_EL_modelfree,feng2016EL_modelfree_tracking}. Rendezvous to a stationary leader with collision avoidance is studied in \cite{hokayem2010coordination_lagrangian}. For directed networks, several results are available. Passivity analysis in \cite{chopra2012outputsync} showed that synchronisation of the velocities (but not of the positions) is achieved on strongly connected directed networks. Rendezvous to a stationary leader and position consensus was studied in \cite{ye2015ELrendezvous} and \cite{ye2016EL_consensus}, respectively, but the papers assumed that the agents did not have a gravitational term in the dynamics. Leader tracking is studied in \cite{feng2016EL_modelfree_tracking} but restrictive assumptions are placed on the leader. Preliminary work by the authors also appeared in \cite{ye2016ELtracking}, and is further analysed below. 

\subsection{Motivation for Model-Independent Algorithms}
Further study of model-independent algorithms is desirable for several reasons. Given a unique Euler-Lagrange equation, determining the minimum number of parameters in an adaptive algorithm is difficult in general \cite{spong2006robot}. Moreover, the adaptive algorithms require knowledge of the exact equation structure; the algorithms can deal with uncertain constant parameters associated with the agent dynamics but are not robust to unmodelled nonlinear agent dynamics. 
Model-independent algorithms are reminiscent of robust controllers, which stand in conceptual contrast to adaptive controllers. Stability and indeed performance is guaranteed given limited knowledge of upper bounds on parameters of the multiagent system, and without use of any attempt to identify these parameters.

As will be shown in this paper, and similarly to \cite{ye2015ELrendezvous,ye2016EL_consensus}, model-independent controllers are exponentially stable, with a computable minimum rate of convergence. Exponentially stable systems are desired over systems which are asymptotically stable, but not exponentially so, because exponentially stable systems offer improved rejection to small amounts of noise and disturbance. Some algorithms \emph{requiring exact knowledge of the Euler-Lagrange equation} have been shown to be exponentially stable \cite{chung2009cooperative,yang2017EL_global_track}. Further, adaptive controllers will yield exponential stability if certain conditions are satisfied, e.g. persistency of excitation. However, the above detailed works using adaptive algorithms have not verified such conditions. 

\subsection{Contributions of this paper}

In this paper, we propose a discontinuous model-independent algorithm that allows a directed network of Euler-Lagrange agents to track a leader with arbitrary trajectory. First, we assume that the network is fixed and contains a directed spanning tree. Then, we relax this assumption to allow for a network with switching/dynamic interactions. In order to achieve stability, a set of scalar control gains must be sufficiently large, i.e. satisfy a set of lower bounding inequalities. These inequalities involve limited knowledge of the bounds on the agent dynamic parameters, limited knowledge of the network topology, and a bound on the initial conditions (which may be arbitrarily large). This last requirement means the algorithm is semi-globally stable; a larger set of allowed initial conditions simply requires recomputing of the control gains. It is also shown that the algorithm is robust to heterogeneous, bounded disturbances for each individual agent. 

We now record the points of contrast between this paper and the previously mentioned existing works. While several results have been listed studying leader tracking, most involving model-independent algorithms have been studied on undirected graphs. Those which do assume directed graphs primarily use \emph{adaptive algorithms}. Most model-independent algorithms on directed networks consider position consensus or rendezvous to a stationary leader; introduction of a moving leader greatly increases the difficulty of the problem due to the complex, nonlinear Euler-Lagrange dynamics. Additionally, the work in \cite{ye2015ELrendezvous,ye2016EL_consensus} did not consider the gravitational term in the agent dynamics. The work \cite{feng2016EL_modelfree_tracking} studies a model-independent leader tracking algorithm on directed graphs, with the restrictive assumption that the leader trajectory is governed by a marginally stable linear time-invariant second order system, and the system matrix known to all agents. A \emph{major contribution of this paper is to allow for any arbitrary leader trajectory which satisfies some mild and reasonable smoothness and boundedness properties}. In addition, \cite{feng2016EL_modelfree_tracking} does not establish an exponential stability property, whereas the algorithm proposed in this paper does. 


A preliminary version of this paper appeared in \cite{ye2016ELtracking}. This paper significantly extends the preliminary version in several aspects. First, we introduce an additional control gain which allows for an additional degree of freedom in selecting the control gains to ensure stability. Moreover, increasing the new gain ensures stability but at the same time, it does not negatively affect convergence rate, unlike in \cite{ye2016ELtracking}. Second, we address the issues arising from the discontinuous nature of the control algorithm by using an approximation of the signum function. An explicit expression relating the tracking error to the degree of approximation and control gain is derived. Additionally, switching topology is considered.  Details of omitted proofs are also now provided.

The paper is structured as follows. Section~\ref{sec:prelim} introduces mathematical preliminaries, and the problem. The problem with fixed network topology, and dynamic topology, is solved in Section~\ref{sec:main} and \ref{sec:switch_top}, respectively. Simulations are provided in Section~\ref{sec:sim} and the paper concluded in Section~\ref{sec:conclusion}.

\section{Background and Problem Statement}\label{sec:prelim}

\subsection{Mathematical Notation and Matrix Theory}\label{subsec:matrix_theory_background}
To begin with, definitions of notation and several results are now provided. The Kronecker product is denoted as $\otimes$.
Denote the $p\times p$ identity matrix as $\mat{I}_p$ and the $n$-column vector of all ones as $\vect 1_n$. The $l_1$-norm and Euclidean norm of a vector $\vect{x}$, and matrix $\mat{A}$, are denoted by $\Vert \cdot \Vert_1$ and \mbox{$\Vert \cdot \Vert_2$}, respectively. 
The signum function is denoted as $\text{sgn}(\cdot)$. For an arbitrary vector $\vect x$, the function $\text{sgn}(\vect x)$ is defined element-wise. A matrix $\mat{A} = \mat{A}^\top$ that is positive definite (respectively nonnegative definite) is denoted by $\mat{A} > 0$ (respectively $\mat A \geq 0$). For two symmetric matrices $\mat{A}, \mat{B}$, $\mat{A} > \mat{B}$ is equivalent to \mbox{$\mat{A} - \mat{B} > 0$}. For a matrix $\mat{A} =\mat{A}^\top$, the minimum and maximum eigenvalues are $\lambda_{\min}(\mat{A})$ and $\lambda_{\max}(\mat{A})$ respectively. The following inequalities hold
\begin{subequations}
\begin{align}
\lambda_{\min}(\mat{A}) &> \lambda_{\max}(\mat{B}) \Rightarrow \mat{A} > \mat{B} \label{eq:lambda_PD} \\
\lambda_{\max}(\mat{A}+\mat{B}) &\leq \lambda_{\max}(\mat{A})+\lambda_{\max}(\mat{B}) \label{eq:lambda_sum_max} \\
\lambda_{\min}(\mat{A}+\mat{B}) & \geq \lambda_{\min}(\mat{A})+\lambda_{\min}(\mat{B}) \label{eq:lambda_sum_min} \\
\lambda_{\min}(\mat{A})\vect{x}^\top\vect{x} &\leq \vect{x}^\top\mat{A}\vect{x}\leq \lambda_{\max}(\mat{A})\vect{x}^\top\vect{x} \label{eq:lambda_xTx}
\end{align}
\end{subequations}

\begin{definition}
A function $f(\vect{x}) : \mathcal{D} \to \mathbb{R}$, where $\mathcal{D} \subseteq \mathbb{R}^n$, is said to be positive definite in $\mathcal{D}$ if $f(\vect{x}) > 0$ for all $\vect{x} \in \mathcal{D}$, except $f(\vect{0}) = 0$.
\end{definition}


\begin{lemma}[The Schur Complement \cite{horn2012matrixbook}]\label{theorem:schur_complement} 
Consider a symmetric block matrix, partitioned as 
\begin{equation}\label{eq:schur_def_A}
\mat{A} = \begin{bmatrix}\mat{B} & \mat{C} \\ \mat{C}^\top & \mat{D}\end{bmatrix}
\end{equation}
Then, $\mat{A} > 0 $  if and only if $\mat{B} > 0$ and $\mat{D} - \mat{C}^\top\mat{B}^{-1}\mat{C} > 0$, or equivalently, if and only if $\mat{D} > 0$ and $\mat{B} - \mat{C}\mat{D}^{-1}\mat{C} > 0$.
\end{lemma}

\begin{lemma}\label{lem:W_bound}
Suppose $\mat A > 0$ is defined as in \eqref{eq:schur_def_A}. Let a quadratic function with arguments $\vect x, \vect y$ be expressed as $W = [\vect x^\top, \vect y^\top] \mat A [\vect x^\top, \vect y^\top]^\top$. Define $\mat F: = \mat B - \mat C \mat D^{-1} \mat C^\top $ and $\mat G := \mat{D}-\mat{C}^\top\mat{B}^{-1}\mat{C}$. Then, there holds
\begin{subequations}\label{eq:W_bound_xy}
\begin{align} \label{eq:W_bound_x}
 \lambda_{\min}(\mat F) \vect x^\top \vect x \leq  \vect{x}^\top\mat F\vect{x}  & \leq W  \\
\lambda_{\min}(\mat G) \vect y^\top \vect y \leq  \vect{y}^\top\mat G\vect{y}  & \leq W \label{eq:W_bound_y} 
\end{align}
\end{subequations}
\begin{proof} We obtain \eqref{eq:W_bound_y} by recalling Lemma~\ref{theorem:schur_complement} and observing that $W = \vect{y}^\top\mat G\vect{y}  + [\vect{y}^\top\mat{C}^\top\mat{B}^{-1}+\vect{x}^\top]\mat{B}[\mat{B}^{-1}\mat{C}\vect{y}+\vect{x}]$.
An equally straightforward proof yields \eqref{eq:W_bound_x}.
\end{proof}
\end{lemma}


\begin{lemma}\label{lem:bounded_pd_function}
Let $g(x,y)$ be a function given as
\begin{equation}
g(x,y) = a x^2 + b y^2 - c x y^2 - d x y
\end{equation} 
for real positive scalars $c,d > 0$. Then for a given $\mathcal{Y} > 0$, there exist $a, b > 0$ such that $g(x,y)$ is positive definite for all $y \in [0, \infty)$ and $x \in [0,\mathcal{X}]$. 
 \begin{proof}
 Observing that $cxy^2 \leq c\mathcal{X}y^2$ for all $x \in [0,\mathcal{X}]$, yields
 \begin{equation}\label{eq:bounded_pd_ineq}
 g(x,y) \geq a x^2 + (b - c\mathcal{X}) y^2 - d x y
 \end{equation} 
 if $y \in [0, \infty)$ and $x \in [0,\mathcal{X}]$ because $c > 0$. For any fixed value of $y = y_1 \in [0,\infty)$, write $\bar g(x) = a x^2 + (b - c\mathcal{X}) y_1^2 - d x y_1$. The discriminant of $\bar g(x)$ is negative if 
 \begin{equation}\label{eq:b_bound_lemma_pd}
 b > c\mathcal{X} + \frac{d^2}{4a}
 \end{equation}
 which implies that the roots of $\bar g(x)$ are complex, i.e. $\bar g(x) > 0$ and this holds for any $y_1 \in [0,\infty)$. We thus conclude that for all $y \in [0, \infty)$ and $x \in [0,\mathcal{X}]$, if $a, b$ satisfies \eqref{eq:b_bound_lemma_pd}, then $g(x,y) > 0$ except the case where $g(x,y) = 0$ if and only if $x = y = 0$. 
 \end{proof}
\end{lemma}

\begin{corollary}\label{cor:bounded_pd_function_2}
Let $h(x,y)$ be a function given as
\begin{equation}
h(x,y) = a x^2 + b y^2 - c x y^2 - d x y - e x - f y 
\end{equation} 
where the real positive scalars $c,d,e,f$ and two further positive scalars $\varepsilon, \vartheta$ are fixed. Suppose that for given $\mathcal{Y}, \mathcal{X}$ there holds $\mathcal{Y}- \varepsilon > 0$, and $\mathcal{X}-\vartheta > 0$. Define the sets $\mathcal{U} = \{x,y : x\in [\mathcal{X}-\vartheta, \mathcal{X}], y > 0\}$ and $\mathcal{V} = \{x, y : x > 0, y\in [\mathcal{Y}-\varepsilon,\mathcal{Y}]\}$. Define the region $\mathcal{R} = \mathcal{U} \cup \mathcal{V}$. Then, there exist $a, b > 0$ such that $h(x,y)$ is positive definite in $\mathcal{R}$.
 \begin{proof}
 Observe that $h(x,y) = g(x,y) - e x -f y$ where $g(x,y)$ is defined in Lemma~\ref{lem:bounded_pd_function}. Let $a^*, b^*$ be such that they satisfy condition \eqref{eq:b_bound_lemma_pd} in Lemma \ref{lem:bounded_pd_function} and thus $g(x,y) > 0$ for $x\in [0,\mathcal{X}]$ and $y\in[0,\infty)$. Note that the positivity condition on $g(x,y)$ in Lemma \ref{lem:bounded_pd_function} continues to hold for any $a\geq a^*$ and any $b \geq b^*$. Let $a_1$ and $b_1$ be positive scalars whose magnitudes will be determined later. Define $a = a_1 + a^*$ and $b = b_1 + b^*$. Define $z(x,y) \triangleq a_1 x^2 + b_1 y^2 - e x - f y$.  Next, consider $(x,\bar y) \in \mathcal{V}$, where $\bar y$ is some fixed value. It follows that
 \begin{equation}
 z(x, \bar y) = a_1 x^2 -  e x + (b_1\bar y^2-f\bar y)
 \end{equation} 
 Note the discriminant of $z(x,\bar y)$ is $\mathcal{D}_x =  e^2 - 4 a_1 ( b_1 \bar y^2 - f \bar y)$. It follows that $\mathcal{D}_x < 0$ if $b_1 \bar y^2 > f\bar y + e^2/4a_1$. This is satisfied, independently of $\bar y \in [\mathcal{Y}-\varepsilon, \mathcal{Y}]$, for any $b_1 \geq b_{1,y}$, $a_1 \geq a_{1,y}$ where 
 \begin{equation}\label{eq:Du_b_ineq}
 b_{1,y} > \frac{ e^2 }{4a_{1,y} (\mathcal{Y}- \varepsilon)^2 } + \frac{f}{\mathcal{Y}-\varepsilon}
 \end{equation}
 because $\mathcal{Y}-\varepsilon \leq \bar y$. It follows that $\mathcal{D}_x < 0 \Rightarrow z(x,y) > 0$ in $\mathcal{V}$. Now, consider $(\bar x,y) \in \mathcal{U}$ for some fixed value $\bar x$. It follows that
 \begin{equation}
 z(\bar x,y) = b_1 y^2 - f y + (a_1 \bar x^2 - e\bar x)
 \end{equation}
 and note the discriminant of $z(\bar x,y)$ is $\mathcal{D}_y = f^2- 4 b_1 (a_1\bar x^2-e\bar x)$. Suppose that $a_1 > e/\mathcal{X}$, which ensures that $a_1\bar x^2-e\bar x > 0$. Then, $\mathcal{D}_y < 0$ if $b_1(a_1\bar x^2 - e\bar x) > f/4$. This is satisfied, independently of $\bar x \in [\mathcal{X} - \vartheta, \mathcal{X}]$, for any $b_1 \geq b_{1,x}$, $a_1 \geq a_{1,x}$ where
 \begin{equation}\label{eq:Dv_b_ineq}
 b_{1,x} > \frac{f}{4(a_{1,x} (\mathcal{X} - \vartheta)^2 - e (\mathcal{X} - \vartheta))}
 \end{equation}
 It follows that $\mathcal{D}_y < 0 \Rightarrow z(x,y) > 0$ in $\mathcal{U}$. We conclude that setting \mbox{$b = b^* + \max [b_{1,x}, b_{1,y}]$} and \mbox{$a = a^* + \max [a_{1,x}, a_{1,y}]$}, implies $h(x, y) > 0$ in $\mathcal{R}$, except $h(0,0) = 0$.
 \end{proof}
\end{corollary}
The results of Lemma~\ref{lem:bounded_pd_function} and Corollary~\ref{cor:bounded_pd_function_2} are almost intuitively obvious.  However, the detailed statements in the proof lay out explicit inequalities for $a,b$. These inequalities will be used to show that for any given Euler-Lagrange network, \emph{control gains can always be found} to ensure leader tracking is achieved.

\subsection{Graph Theory}\label{subsec:graph_theory}
The agent interactions can be modelled by a weighted directed graph which is denoted as $\mathcal{G} = (\mathcal{V,E,A})$, with the set of nodes $\mathcal{V} = \{v_0,v_1,\ldots,v_n\}$, and with a corresponding set of ordered edges $\mathcal{E} \subseteq \mathcal{V \times V}$. 
A directed edge $e_{ij} = (v_i, v_j)$ is outgoing with respect to $v_i$ and incoming with respect to $v_j$, and implies that $v_j$ is able to obtain some information from $v_i$. The precise nature of this information will be made clear in the sequel. The weighted adjacency matrix $\mathcal{A}\in\mathbb{R}^{(n+1)\times (n+1)}$ of $\mathcal{G}$ has nonnegative elements $a_{ij}$. The elements of $\mathcal{A}$ have properties such that $a_{ij} > 0 \Leftrightarrow e_{ji}\in \mathcal{E}$ while $a_{ij} = 0$ if $e_{ji} \notin \mathcal{E}$ and it is assumed $a_{ii} = 0,  \forall i$. The neighbour set of $v_i$ is denoted by $\mathcal{N}_i = \{v_j \in \mathcal{V} : (v_j, v_i) \in \mathcal{E} \}$. The $(n+1)\times (n+1)$ Laplacian matrix, $\mathcal{L} = \{l_{ij}\}$, of the associated digraph $\mathcal{G}$ is defined as $l_{ij} = -a_{ij}$ for $j\neq i$ and $l_{ij} =  \sum^n_{k=1,k\neq i} \:a_{ik}$ if $j=i$. 
A directed spanning tree is a directed graph formed by directed edges of the graph that connects all the nodes, and where every vertex apart from the root has exactly one parent. A graph is said to contain a spanning tree if a subset of the edges forms a spanning tree. We make use of the following standard lemma.

\begin{lemma}[\cite{zhang2015graph_lyapunov}]\label{cor:dst_laplacian_pd}
Let $\mathcal{G}$ contain a directed spanning tree, and suppose there are no edges which are incoming to the root vertex of the tree, which without loss of generality, is set as $v_0$. Then the Laplacian of $\mathcal{G}$ can be partitioned as
\begin{equation}\label{eq:laplacian_partition}
\mathcal{L} = \begin{bmatrix} 0 & \mat{0} \\ \mathcal{L}_{11} & \mathcal{L}_{22}\end{bmatrix}
\end{equation}
and $\exists\, \mat{\Gamma} > 0$ which is diagonal and $\mat{\Gamma}\mathcal{L}_{22}+\mathcal{L}^{\top}_{22}\mat{\Gamma} > 0$.
\end{lemma}
For future use, denote the $i^{th}$ diagonal element of $\mat{\Gamma}$ as $\gamma_i$ and define $\bar\gamma \triangleq \max_i \gamma_i$ and $\underline{\gamma} \triangleq \min_i \gamma_i$.

\subsection{Euler-Lagrange Systems}\label{subsec:el_system_background}
The $i^{th}$ Euler-Lagrange agent's equation of motion is:
\begin{equation}\label{eq:euler_lagrange_general_agent}
\mat{M}_i(\vect{q}_i)\ddot{\vect{q}}_i + \mat{C}_i(\vect{q}_i,\dot{\vect{q}}_i)\dot{\vect{q}}_i + \vect{g}_i(\vect{q}_i) + \vect{\zeta}_i = \vect{\tau}_i 
\end{equation}
where $\vect{q}_i(t)\in \mathbb{R}^p$ is a vector of the generalised coordinates. Note that from here onwards, we drop the time argument $t$ whenever there is no ambiguity. The inertia matrix is $\mat{M}_i(\vect{q}_i) \in \mathbb{R}^{p\times p}$, $\mat{C}_i(\vect{q}_i,\dot{\vect{q}}_i)\in\mathbb{R}^{p\times p}$ is the Coriolis and centrifugal force matrix, $\vect{g}_i\in\mathbb{R}^{p}$ is the vector of (gravitational) potential forces and $\vect{\zeta}_i(t)$ is an unknown, time-varying disturbance. It is assumed that all agents are fully-actuated, with $\vect{\tau}_i\in\mathbb{R}^p$ being the control input vector. For each agent, the $k^{th}$ generalised coordinate is denoted using superscript $(k)$; thus $\vect{q}_i = [q^{(1)}_i, ..., q^{(p)}_i]^\top$. It is assumed that the systems described using \eqref{eq:euler_lagrange_general_agent} have the following properties given below: 
\begin{enumerate}[label=P\arabic*]
\item \label{prty:P1_M_PD} The matrix $\mat{M}_i(\vect{q}_i)$ is symmetric positive definite. 
\item \label{prty:P2_M_bound} There exist scalar constants $k_{\underline{m}}, k_{\overline{M}} > 0$ such that $k_{\underline{m}}\mat{I}_p\leq \mat{M}_i(\vect{q}_i)\leq k_{\overline{M}}\mat{I}_p, \forall\,i,\vect q_i$. It follows that $\sup_{\vect{q}_i}\Vert\mat{M}_i\Vert_2 \leq k_{\overline{M}}$ and $k_{\underline{m}}\leq \inf_{\vect{q}_i}\Vert\mat{M}_i^{-1}\Vert_2$ holds $\forall\,i$.
\item \label{prty:P3_skew} The matrix $\mat{C}_i(\vect{q}_i,\dot{\vect{q}}_i)$ is defined such that $\dot{\mat{M}}_i-2\mat{C}_i$ is skew-symmetric, i.e. $\dot{\mat{M}}_i = \mat{C}_i+\mat{C}_i^\top$. 
\item \label{prty:P4_Cg_bound} There exist scalar constants $k_{C}, k_g > 0$ such that \mbox{$\Vert\mat{C}_i\Vert_2 \leq k_{C} \Vert\dot{\vect{q}}_i\Vert_2,\forall\,i,\dot{\vect q}_i$} and \mbox{ $\Vert\vect g_i\Vert_2 < k_{g},\forall i$}.
\item \label{prty:P5_dist_bound} There exists a constant $k_{\zeta}$ such that \mbox{$\Vert \vect{\zeta}_i \Vert_2 \leq k_{\zeta}, \forall\, i$.}
\end{enumerate}
Properties \ref{prty:P1_M_PD}-\ref{prty:P4_Cg_bound} are standard and widely assumed properties of Euler-Lagrange dynamical systems, see \cite{ortega2013passivity_book} for details. Property \ref{prty:P5_dist_bound} is a reasonable assumption on disturbances.

\subsection{Problem Statement}
The leader is denoted as agent 0, i.e. vertex $v_0$, with $\vect{q}_0$(t) and $\dot{\vect{q}}_0(t)$ being its time-varying generalised coordinates and generalised velocity, respectively. The objective is to develop a model-independent, distributed algorithm which allows a directed network of Euler-Lagrange agents to synchronise and track the trajectory of the leader. The leader tracking objective is said to be achieved if $\lim_{t\to \infty}\Vert\vect{q}_i(t)-\vect{q}_0(t)\Vert_2 = 0$ and $\lim_{t\to\infty}\Vert\dot{\vect{q}}_i(t) -\dot{\vect{q}}_0(t)\Vert_2 = 0$ for all $i = 1, \hdots, n$. By model-independent, we mean that the algorithm does not contain $\mat{M}_i, \mat{C}_i, \vect g_i\,\forall\, i$ nor make use of an associated linear parametrisation.
Two mild assumptions are now given.
\begin{assumption}\label{assm:leader_traj} The leader trajectory $\vect{q}_0(t)$ is a $\mathcal{C}^2$ function with derivatives $\dot{\vect{q}}_0$ and $\ddot{\vect{q}}_0$ which are bounded as \mbox{$\Vert\vect{1}_n\otimes \dot{\vect{q}}_0\Vert_2 \leq k_p$} and $\Vert\vect{1}_n\otimes \ddot{\vect{q}}_0\Vert_2 \leq k_q$. The positive constants $k_p, k_q$ are known a priori.
\end{assumption}

\begin{assumption}\label{assm:IC}
All possible initial conditions lie in some fixed but arbitrarily large set that is known. In particular, $\Vert \vect{q}_i \Vert_2 \leq k_a/\sqrt{n}$ and $\Vert \dot{\vect{q}}_i \Vert_2 \leq k_b/\sqrt{n}$, where $k_a, k_b$ are known a priori. 
\end{assumption}
These two assumptions are not unreasonable, as many systems will have an expected operating range for $\vect{q}$ and $\dot{\vect{q}}$.

The follower agents' capability to \emph{sense} relative states is captured by the directed graph $\mathcal{G}_A$ with an associated Laplacian $\mathcal{L}_A$. In Section~\ref{sec:main}, we assume $\mathcal{G}_A$ is fixed. Later in Section~\ref{sec:switch_top}, it is assumed that $\mathcal{G}_A$ is dynamic, i.e. time-varying. 
Thus, if $a_{ij}>0$ then agent $i$ 
can sense $\vect{q}_i-\vect{q}_j$ and $\dot{\vect{q}}_i-\dot{\vect{q}}_j$. We denote the neighbour set of agent $i$ on $\mathcal{G}_A$ as $\mathcal{N}_{Ai}$. We further assume that agent $i$ can measure its own $\vect{q}_i$ and $\dot{\vect q}_i$. A second weighted and directed time-varying graph $\mathcal{G}_B(t)$, with the associated Laplacian $\mathcal{L}_B(t)$, exists between the followers to \emph{communicate} estimates of the leader's state. Denote the neighbour set of agent $i$ on $\mathcal{G}_B(t)$ at time $t$ as $\mathcal{N}_{Bi}(t)$. Note that $v_j\in\mathcal{N}_{Bi}(t)$ when agent $j$ communicates  directly to agent $i$ its estimates of the leader's state at time $t$ (the precise nature of this estimate is described in Section~\ref{subsec:distributed_observer}). Further note that $\mathcal{G}_A$ is not necessarily equal to $\mathcal{G}_B$ and so $\mathcal{N}_{Ai} \neq\mathcal{N}_{Bi}$ in general. However the node sets of $\mathcal{G}_A$ and $\mathcal{G}_B$ are the same.


\begin{remark}[Comparison of this paper to recent leader tracking results]
Almost all mechanical systems will have trajectories which satisfy the mild Assumption~\ref{assm:leader_traj}. In comparison, more restrictive assumption are made on the leader trajectory in \cite{abdessameud2016TAC_eltrack,feng2016EL_modelfree_tracking}. In   \cite{abdessameud2016TAC_eltrack,feng2016EL_modelfree_tracking}, the leader trajectory is describable by an LTI system, with system matrix defined as $\mat{S}$. In \cite{abdessameud2016TAC_eltrack}, it is assumed that all eigenvalues of $\mat{S}$ are purely imaginary. In \cite{feng2016EL_modelfree_tracking}, it is assumed that $\mat{S}$ is marginally stable. More importantly, both \cite{abdessameud2016TAC_eltrack} and \cite{feng2016EL_modelfree_tracking} assume that $\mat{S}$ is \textbf{known to all agents}, which is a highly restrictive assumption. As will become apparent in the sequel, we use a distributed observer to allow every agent to obtain $\vect{q}_0(t)$ and $\dot{\vect{q}}_0(t)$ precisely. The work \cite{ghapani2016el_flocking_leader} has similar assumptions to this paper, but uses an adaptive algorithm and is therefore fundamentally different to the model-independent controller studied in this paper.
\end{remark}

\section{Leader Tracking on Fixed Directed Networks}\label{sec:main}

\subsection{Finite-Time Distributed observer}\label{subsec:distributed_observer}
Before we show the main result, we detail a distributed finite-time observer developed in \cite{cao2010SL_estimators} which allows each follower agent to obtain $\vect q_0$ and $\dot{\vect q}_0$. Let $\wh{\vect r}_i$ and $\wh{\vect v}_i$ be the $i^{th}$ agent's estimated values for the leader position and velocity respectively. Agent $i\in \{1, \hdots, n\}$ runs the observer
\begin{subequations}\label{eq:dist_observer}
\begin{align}
\dot{\wh{\vect r}}_i & = \wh{\vect v}_i - \omega_1 \sgn\big(\sum_{j\in\mathcal{N}_{Bi}(t)} b_{ij}(t) (\wh{\vect r}_i - \wh{\vect r}_j)\big) \label{eq:observer_pos} \\
\dot{\wh{\vect v}}_i & = -\omega_2 \sgn\big(\sum_{j\in\mathcal{N}_{Bi}(t)} b_{ij}(t) (\wh{\vect v}_i - \wh{\vect v}_j)\big) \label{eq:observer_vel}
\end{align}
\end{subequations}
where $b_{ij}$ are the elements of the adjacency matrix associated with graph $\mathcal{G}_B(t)$ and $\omega_1, \omega_2 > 0$ are internal gains of the observer. Clearly, if $a_{i0} > 0$ then \mbox{agent $i$} can directly sense the leader, $v_0$ and thus learns of $\vect{q}_0$ and $\dot{\vect q}_0$. For such an \mbox{agent $i$}, we set $b_{i0} > 0$ and $\wh{\vect r}_0(t) = \vect q_0(t)$ and $\wh{\vect v}_0(t) = \dot{\vect q}_0(t)$; agent $i$ still runs the distributed observer \eqref{eq:dist_observer}. We now give a theorem for convergence of the observer, and explain below why all followers execute \eqref{eq:dist_observer} even if they learn of $\vect{q}_0$, $\dot{\vect q}_0$ from $\mathcal{G}_A$. 
\begin{theorem}[Theorem 4.1 of \cite{cao2010SL_estimators}]\label{thm:dist_observer_cao}
Suppose that the leader trajectory $\vect{q}_0(t)$ satisfies Assumption~\ref{assm:leader_traj}. If at every $t$, $\mathcal{G}_B(t)$ contains a directed spanning tree, and $\omega_2 > k_q/n$ then, for some $T_1 < \infty$, there holds $\wh{\vect r}_i(t) = \vect q_0(t)$ and $\wh{\vect v}_i(t) = \dot{\vect q}_0(t)$ for all $i\in \{1, \hdots, n\}$, for all $t\geq T_1$.
\end{theorem}

The key reason for agent $i$ to run the distributed observer even if $a_{i0} > 0$ (and thus agent $i$ knows $\vect{q}_0$ and $\dot{\vect q}_0$) is to ensure robustness to network changes over time (e.g. switching topology due to loss of connection). We elaborate further. In the case of a fixed $\mathcal{G}_A$, then agent $i$ will know $\vect{q}_0$ and $\dot{\vect q}_0$ for all $t$ and there will be no need for the observer. However, we explore switching $\mathcal{G}_{A}(t)$ in Section~\ref{sec:switch_top}. Consider the case where $a_{i0}(t) = 1$ for $t\in [0,10)$ and $a_{i0}(t) = 0$ for $t\in [10,\infty)$. If agent $i$ does not run \eqref{eq:dist_observer}, then for $t\geq 10$, it would not know $\vect{q}_0(t)$ and $\dot{\vect q}_0(t)$ because $a_{i0}(t) = 0\Rightarrow b_{i0}(t) = 0$. If agent $i$ runs \eqref{eq:dist_observer} from $t=0$, then for all $t \geq T_1$, it is guaranteed that $\wh{\vect r}_1(t) = \vect q_0(t)$ and $\wh{\vect v}_1(t) = \dot{\vect q}_0(t)$ even if $\mathcal{G}_A(t)$ switches, so long as the connectivity condition in Theorem~\ref{thm:dist_observer_cao} is satisfied. A second reason is that agent $i$ may acquire states by sensing over $\mathcal{G}_A$; \eqref{eq:dist_observer} acts as a filter for noisy measurements.




\subsection{Model-Independent Control Law}

Consider the following  algorithm for the $i^{th}$ agent
\begin{align}\label{eq:claw_track_sign}
\tau_i & = - \eta\sum_{j\in\mathcal{N}_{Ai}} a_{ij}\Big( (\vect{q}_i-\vect{q}_j) +\mu(\dot{\vect{q}}_i-\dot{\vect{q}}_j) \Big) \nonumber \\
& \qquad - \beta\sgn\left((\vect q_i- \wh{\vect r}_i)+\mu(\dot{\vect q}_i- \wh{\vect v}_i)\right)
\end{align}
where $a_{ij}$ is the weighted $(i,j)$ entry of the adjacency matrix $\mathcal{A}$ associated with the weighted directed graph $\mathcal{G}_A$. The control gains $\mu, \eta$ and $\beta$ are strictly positive constants and their design will be specified later. For simplicity, it is assumed that $\eta > 1$. Note that for all $i$, for all $t > T_1$, $\wh{\vect r}_i$ is replaced with $\vect q_0$ and $\wh{\vect v}_i$ replaced with $\dot{\vect q}_0$. 

Let us denote the new error variable $\widetilde{\vect{q}}_i = \vect{q}_i - \vect{q}_0$. Let $\widetilde{\vect{q}} = [\widetilde{\vect{q}}_1^\top, ..., \wt{\vect{q}}_n^\top]^\top \in \mathbb{R}^{np\times 1}$ the stacked column vector of all $\wt{\vect{q}}_i$. The leader tracking objective is therefore achieved if $\widetilde{\vect{q}}(t) = \dot{\widetilde{\vect{q}}}(t) = 0$ as  $t\to\infty$. We denote $\vect{g} = [\vect{g}_1^\top, ..., \vect{g}_n^\top]^\top$, $\vect{\zeta} = [\vect{\zeta}_1^\top, ..., \vect{\zeta}_n^\top]^\top$, $\vect{q} = [\vect q _1^\top, ..., \vect q _n^\top]^\top$, and $\dot{\vect{q}} = [\dot{\vect q} _1^\top, ..., \dot{\vect q} _n^\top]^\top$ as the $np\times 1$ stacked column vectors of all $\vect{g}_i, \vect{\zeta}_i, \vect{q}_i$ and $\dot{\vect q}_i$ respectively. Let $\mat{M}(\vect{q}) = diag[\mat{M}_1(\vect{q}_1), ..., \mat{M}_n(\vect{q}_n)] \in \mathbb{R}^{np\times np}$, and $\mat{C}(\vect{q},\dot{\vect{q}}) = diag[\mat{C}_1(\vect{q}_1, \dot{\vect{q}}_1), ..., \mat{C}_n(\vect{q}_n, \dot{\vect{q}}_n)] \in \mathbb{R}^{np\times np}$.  Since $\mat{M}_i > 0,\,\forall\,i$ then $\mat{M}$ is also symmetric positive definite. Define an error vector, $\vect e_i = \wh{\vect r}_i - \vect q_0, \forall i = 1,...,n$ and $\dot{\vect e}_i = \wh{\vect v}_i - \dot{\vect q}_0$.  Define $\vect e = [\vect e_1^\top, ..., \vect e_n^\top]^\top \in \mathbb{R}^{np\times 1}$, $\dot{\vect e} = [\dot{\vect e}_1^\top, ..., \dot{\vect e}_n^\top]^\top \in \mathbb{R}^{np\times 1}$. 

The definition of $\wt{\vect{q}}_i$ yields $\mat{M}_i\ddot{\widetilde{\vect{q}}}_i = \mat{M}_i\ddot{\vect{q}}_i - \mat{M}_i\ddot{\vect{q}}_0$ and combining the agent dynamics \eqref{eq:euler_lagrange_general_agent} and the control law \eqref{eq:claw_track_sign}, the closed-loop system for the follower network, with nodes $v_1, \hdots, v_n$, can be expressed as 
\begin{align}\label{eq:network_system}
&\ddot{\wt{\vect{q}}} \in ^{a.e.} \mathcal{K} \big[ -\mat{M}^{-1}[ \mat{C}\dot{\wt{\vect{q}}} + \eta(\mathcal{L}_{22}\otimes\mat{I}_p)(\wt{\vect{q}}+ \mu\dot{\wt{\vect{q}}}) + \vect{g} + \vect{\zeta}\nonumber\\
&  +\beta\sgn\left(\vect s +\mu\dot{\vect s}\right) + \mat{M}(\vect{1}_n\otimes\ddot{\vect{q}}_0) + \mat{C}(\vect{1}_n\otimes\dot{\vect{q}}_0)]\big]
\end{align} 
where $\mathcal{K}$ denotes the differential inclusion, $a.e.$ stands for ``almost everywhere" and $\vect s = \wt{\vect q} - \vect e$. Here, $\mathcal{L}_{22}$ is the lower block matrix of $\mathcal{L}_A$ as partitioned in \eqref{eq:laplacian_partition}. Filippov solutions of $\wt{\vect q}$ and $\dot{\wt{\vect q}}$ for \eqref{eq:network_system} exist because the signum function is measurable and locally essentially bounded, and $\wt{\vect q}$ and $\dot{\wt{\vect q}}$ are absolutely continuous functions of time \cite{cortes2008discontinuous_tutorial}.


\subsection{An Upper Bound Using Initial Conditions}\label{subsec:IC_bound}
Before proceeding with the main proof, we calculate an upper bound (which may not be tight) on the initial states expressed as $\Vert \wt{\vect q}(0) \Vert_2 < \mathcal{X}$ and $\Vert \dot{\wt{\vect q}}(0) \Vert_2 < \mathcal{Y}$ using Assumption~\ref{assm:IC}. In the sequel, we show that these bounds hold for all time, and exponential convergence results.
In keeping with the model-independent approach, define 
\begin{equation}\label{eq:V_bar}
\bar V_\mu = \begin{bmatrix}\wt{\vect{q}} \\ \dot{\wt{\vect{q}}}\end{bmatrix}^\top
\begin{bmatrix}
\eta\lambda_{\max}(\mat{X})\mat{I}_{np} & \frac{1}{2}\mu^{-1} (k_{\overline{M}}+\delta)\mat{I}_{np} \\ \frac{1}{2}\mu^{-1} (k_{\overline{M}}+\delta)\mat{I}_{np} & \frac{1}{2}(k_{\overline{M}}+\delta) \mat{I}_{np}
\end{bmatrix} 
\begin{bmatrix}\wt{\vect{q}} \\ \dot{\wt{\vect{q}}}\end{bmatrix}
\end{equation}
where\footnote{Note that $\bar V_\mu$ is not a Lyapunov function}, $\mat{X} = (\mat{\Gamma}\mathcal{L}_{22}+\mathcal{L}_{22}^\top\mat{\Gamma})\otimes\mat{I}_p > 0$ from Lemma~\ref{cor:dst_laplacian_pd} and the constant $\delta > 0$ is sufficiently small such that $k_{\underline{m}}-\delta > 0$. Without loss of generality, assume that $\mat{\Gamma}$ is scaled such that $\bar\gamma = 1$. Let the matrix in \eqref{eq:V_bar} be $\mat L_\mu$. By observing that \mbox{$(k_{\overline{M}}+\delta)\mat{I}_{np} > \mat M$}, then according to Lemma~\ref{theorem:schur_complement}, $\mat{L}_\mu > 0$ if and only if $\eta\lambda_{\max}(\mat{X})\mat{I}_{np} - \frac{1}{2}\mu^{-2} (k_{\overline{M}} + \delta)\mat{I}_{np} > 0$. It follows that $\mat{L}_\mu > 0$ for any \mbox{$\mu\geq \mu_1^*$ where $\mu_1^* > \sqrt{(k_{\overline{M}} + \delta)/2\lambda_{\max}(\mat{X})}$.} Since $\mat{X} > 0$, such a $\mu_1^*$ always exists. For convenience, we use $\bar V_\mu(t)$ to denote $\bar V_\mu (\wt{\vect q}(t), \dot{\wt{\vect q}}(t))$, and observe that there holds
\begin{align}\label{eq:Vbar_upper}
\bar{V}_\mu(t) & \!\leq\! \eta\lambda_{\max}(\mat{X}) {\Vert \wt{\vect{q}} \Vert_2^2} \!+\! (k_{\overline{M}} \!+\! \delta)\big(\frac{1}{2} {\Vert \dot{\wt{\vect q}} \Vert_2^2} \!+\! \mu^{-1} \Vert \wt{\vect{q}} \Vert_2 \Vert \dot{\wt{\vect q}} \Vert_2 \big)
\end{align}
for all $t$. Next, define
\begin{equation}\label{eq:V_LB}
\underline{V}_{\mu} \!=\! \frac{1}{2}\begin{bmatrix}\wt{\vect{q}} \\ \dot{\wt{\vect{q}}}\end{bmatrix}^\top\!\!
\begin{bmatrix}
\frac{1}{2}\eta\lambda_{\min}(\mat{X})\mat{I}_{np} & \hspace*{-9pt} \mu^{-1} \underline{\gamma} (k_{\underline{m}}\!-\!\delta) \mat{I}_{np} \\ \mu^{-1} \underline{\gamma} (k_{\underline{m}}\!-\!\delta) \mat{I}_{np} & \hspace*{-9pt} \underline{\gamma}(k_{\underline{m}}\!-\!\delta) \mat{I}_{np}
\end{bmatrix} \!\!
\begin{bmatrix}\wt{\vect{q}} \\ \dot{\wt{\vect{q}}}\end{bmatrix}
\end{equation}
Call the matrix in \eqref{eq:V_LB} $\mat N_\mu$. Similarly to above, use Lemma~\ref{theorem:schur_complement} to show that $\mat N_\mu > 0$ for any $\mu \geq \mu_2^*$ where $\mu_2^* > \sqrt{2\underline{\gamma}(k_{\underline{m}} - \delta)/\lambda_{\min}(\mat{X})}$. Set $\mu^*_3 = \max\{\mu^*_1, \mu^*_2\}$. Define 
\begin{subequations}
\begin{align}
\rho_1(\mu) & = \eta\lambda_{\max}(\mat X) - \frac{1}{2}\mu^{-2}({k_{\overline{M}}} + \delta) \\
\rho_2(\mu) & = \eta\frac{1}{4} \lambda_{\min}(\mat X)  - \frac{1}{2} \mu^{-2}\underline{\gamma}(k_{\underline{m}}-\delta)
\end{align} 
\end{subequations}
and verify that $\rho_1(\mu_3^*) > \rho_2(\mu_3^*)$. Note that for any $\mu\geq \mu^*_3$ there holds $\bar V_\mu \leq \bar V_{\mu^*_3}$ and $\rho_i(\mu^*_3) \leq \rho_i(\mu), i = 1,2$. Compute
\begin{align*}\label{eq:Vbar_star_upper}
\bar{V}^* = \eta\lambda_{\max}(\mat{X}) k_a^2 + \frac{1}{2}(k_{\overline{M}} + \delta) k_b^2 + {\mu_3^*}^{-1}(k_{\overline{M}} + \delta ) k_a k_b
\end{align*}
 From Assumption~\ref{assm:IC}, one has that ${\Vert \wt{\vect q}(0)\Vert_2} \leq k_a$ and ${\Vert \dot{\wt{\vect q}}(0)\Vert_2} \leq k_b$. Thus, one concludes from \eqref{eq:Vbar_upper} and the above equation that there holds $\bar{V}_\mu(0) \geq \bar{V}^*$ for any $\mu \geq \mu_3^*$. Because we assumed $\eta > 1$, it follows from Lemma~\ref{lem:W_bound} and \eqref{eq:W_bound_x} that 
\begin{equation}\label{eq:initial_X1}
\Vert \wt{\vect q}(0)\Vert_2 \leq \sqrt{\frac{\bar V_{\mu}(0)}{\rho_1(\mu)}} \leq \sqrt{\frac{\bar V_{\mu}(0)}{\rho_1(\mu^*_3)}} < \sqrt{\frac{\bar V^*(0)}{\rho_2(\mu^*_3)}}  \triangleq \mathcal{X}_1
\end{equation}  
Following a similar method yields $\mathcal{Y}_1$. Next, compute
\begin{equation*}
\wh{V}^* = \eta\lambda_{\max}(\mat X){\mathcal{X}_1}^2 + \frac{1}{2} (k_{\overline{M}} +\delta){\mathcal{Y}_1}^2 + (\mu^*_3)^{-1}(k_{\overline{M}} +\delta)\mathcal{X}_1\mathcal{Y}_1
\end{equation*}
and observe that $\bar V^* \leq \wh{V}^*$. Lastly, compute the bound
\begin{equation}\label{eq:X_bound}
\mathcal{X} = \sqrt{\wh{V}^*/\rho_2(\mu^*_3)} 
\end{equation}
and notice that $\Vert \wt{\vect q}(0)\Vert_2 \leq \mathcal{X}_1 \leq \mathcal{X}$. Similarly, $\mathcal{Y}$ is obtained using \eqref{eq:W_bound_y}, with the steps omitted due to spatial limitations. Because both sides of \eqref{eq:X_bound} are independent of $\mu$, the values $\mathcal{Y}$ and $\mathcal{X}$ do not change for all $\mu \geq \mu^*_3$.

\subsection{Stability Proof}

\begin{theorem}\label{theorem:main_result}
Suppose that the conditions in Theorem~\ref{thm:dist_observer_cao} are satisfied. Under Assumptions~\ref{assm:leader_traj} and \ref{assm:IC}, the leader-tracking is achieved exponentially fast if 1) the network $\mathcal{G}_A$ contains a directed spanning tree with the leader as the root node, and 2) the control gains $\mu, \eta, \beta$ satisfy a set of lower bounding inequalities\footnote{In Remark~\ref{rem:gain_design}, we detail an approach for designing the gains.}. For a given $\mathcal{G}_A$ containing a directed spanning tree, there always exists $\mu, \eta, \beta$ which satisfy the inequalities.
\begin{proof}
The proof will be presented in four parts. In \emph{Part 1}, we study a Lyapunov-like candidate function $V$. In \emph{Part 2}, we analyse $\dot{V}$ and show that it is upper bounded. \emph{Part 3} shows that the system trajectory remains bounded for all time, and exponential convergence is proved in \emph{Part 4}.

\emph{Part 1:} Consider the Lyapunov-like candidate function
\begin{equation*}\label{eq:V_01}
V = \frac{1}{2}\eta\wt{\vect{q}}^\top\mat{X}\wt{\vect{q}} + \mu^{-1}\wt{\vect{q}}^\top\mat{\Gamma}_p\mat{M}\dot{\wt{\vect{q}}} + \frac{1}{2}\dot{\wt{\vect{q}}}^\top\mat{\Gamma}_p \mat{M}\dot{\wt{\vect{q}}} = V_1 + V_2 + V_3
\end{equation*}
with $\mat{X}$ given below \eqref{eq:V_bar}, and $\mat{\Gamma}_p = \mat{\Gamma} \otimes \mat{I}_p$. Observe that
\begin{equation}\label{eq:V_01_quad}
V = \begin{bmatrix}\wt{\vect{q}} \\ \dot{\wt{\vect{q}}}\end{bmatrix}^\top
\begin{bmatrix}
\frac{1}{2}\eta\mat{X} & \frac{1}{2}\mu^{-1}\mat{\Gamma}_p\mat{M} \\ \frac{1}{2}\mu^{-1}\mat{\Gamma}_p\mat{M} & \frac{1}{2}\mat{\Gamma}_p\mat{M}
\end{bmatrix} 
\begin{bmatrix}\wt{\vect{q}} \\ \dot{\wt{\vect{q}}}\end{bmatrix}
\end{equation}
Call the matrix in \eqref{eq:V_01_quad} $\mat H_\mu$. 
From Lemma~\ref{theorem:schur_complement}, and the assumed properties of $\mat{M}_i$, there holds $\mat{H}_\mu > 0$ if and only if $\eta\mat{X} - \mu^{-2}\mat{\Gamma}_p\mat{M} > 0$, which is implied by $\lambda_{\min}(\mat{X}) - \mu^{-2}k_{\overline{M}} > 0$. This is because $k_{\overline{M}} \geq \sup_{\vect q} \lambda_{\max}(\mat{M})$, and we assumed that $\eta > 1$ and $\bar\gamma = 1$. For any $\mu \geq \mu_4^*$, where
\begin{equation}\label{eq:cond_V_PD}
\mu_4^* > \sqrt{2k_{\overline{M}}/\lambda_{\min}\left(\mat{X}\right)} 
\end{equation}
there holds $\mat L_\mu> \mat H_\mu > \mat{N}_\mu > 0$ because $\mu_4^* \geq \mu_3^*$ as defined below \eqref{eq:V_LB}. Thus, although the eigenvalues $\lambda_i(\mat{H}_\mu)$ depend on $\vect{q}(t)$, there holds $\lambda_{\min}(\mat{N}_\mu) \leq \lambda_i(\mat{H}_\mu) \leq \lambda_{\max}(\mat{L}_\mu)$ for all $i$, and for all $t\geq 0$. Thus, for any $\mu \geq \mu_4^*$, $V > 0$ and is radially unbounded.
For simplicity, let $V(t)$ denote $V(\wt{\vect q}(t), \dot{\wt{\vect q}}(t))$ and observe that $V(t) < \bar V_{\mu}(t),\forall\, t$ because
\begin{align}\label{eq:V_upper_bound}
V(t) & \leq \frac{1}{2}\eta\lambda_{\max}(\mat X) {\Vert \wt{\vect q}(t)\Vert_2}^2 + \frac{1}{2}k_{\overline{M}}{\Vert \dot{\wt{\vect q}}(t)\Vert_2}^2 \nonumber \\
& \quad + \mu^{-1} k_{\overline{M}}\Vert\wt{\vect q}(t)\Vert_2 \Vert \dot{\wt{\vect q}}(t)\Vert_2
\end{align}

\emph{Part 2:} Let $\dot{V}$ be the set-valued derivative of $V$ with respect to time, along the trajectories of the system \eqref{eq:network_system}. From \eqref{eq:V_01} we obtain $\dot{V} = \dot{V}_1 + \dot{V}_2 + \dot{V}_3$. We obtain $\dot{V}_1 = \eta\wt{\vect{q}}^\top\mat{X}\dot{\wt{\vect{q}}}$. The second summand yields
\begin{equation*}
\dot{V}_2 \in \mu^{-1} \dot{\wt{\vect q}}^\top \mat{\Gamma}_p \mat{M}\dot{\wt{\vect q}} + \mu^{-1}\wt{\vect q}^\top \mat{\Gamma}_p \dot{\mat{M}}\dot{\wt{\vect q}} + \mu^{-1}\wt{\vect q}^\top \mat{\Gamma}_p \mat{M} \times \mathcal{K}\Big[\ddot{\wt{\vect q}}\Big]
\end{equation*}
Substituting $\ddot{\wt{\vect q}}$ from \eqref{eq:network_system}, and using Assumption \ref{prty:P3_skew}, we obtain
\begin{align}\label{eq:V2_dot}
\dot{V}_2 & \in \mathcal{K}\Big[  -\wt{\vect q}^\top(\mat{\Gamma}_p\mathcal{L}_{22}\otimes\mat{I}_p)(\mu^{-1}\eta\wt{\vect{q}} + \eta\dot{\wt{\vect{q}}}) + \mu^{-1} \dot{\wt{\vect q}}^\top \mat{\Gamma}_p \mat{M}\dot{\wt{\vect q}} \nonumber \\
&\; + \mu^{-1}\wt{\vect q}^\top \mat{\Gamma}_p \mat{C}^\top\dot{\wt{\vect q}} - \mu^{-1}\wt{\vect q}^\top \mat{\Gamma}_p (\Delta + \mat{C}(\vect{1}_n\times\dot{\vect{q}}_0)) \nonumber \\
& \quad- \beta\mu^{-1}\wt{\vect q}^\top \mat{\Gamma}_p \sgn\left(\vect s+\mu\dot{\vect s}\right)\Big]
\end{align}
where $\Delta = \vect g +\vect{\zeta} +\mat{M}(\vect{1}_n\times\ddot{\vect{q}}_0)$. Similarly, $\dot{V}_3$ is
\begin{align}
\dot{V}_3 & \in \dot{\wt{\vect q}}^\top \mat{\Gamma}_p \mat{M}\times \mathcal{K} \Big[\ddot{\wt{\vect q}}\Big] + \frac{1}{2}\dot{\wt{\vect{q}}}^\top \mat{\Gamma}_p \dot{\mat{M}}\dot{\wt{\vect{q}}}
\end{align}
Substituting $\ddot{\wt{\vect q}}$ from \eqref{eq:network_system} and using Assumption \ref{prty:P3_skew} we obtain
\begin{align}\label{eq:V3_dot}
& \dot{V}_3  \in \mathcal{K}\Big[ -\eta\dot{\wt{\vect q}}^\top(\mat{\Gamma}\mathcal{L}_{22}\otimes\mat{I}_p)\wt{\vect{q}} - \mu\eta\dot{\wt{\vect q}}^\top(\mat{\Gamma}\mathcal{L}_{22}\otimes\mat{I}_p)\dot{\wt{\vect{q}}} \nonumber \\
& \;\;\;- \beta\dot{\wt{\vect q}}^\top \mat{\Gamma}_p \sgn\left(\vect s +\mu\dot{\vect s}\right)  - \dot{\wt{\vect q}}^\top \mat{\Gamma}_p (\Delta + \mat{C}(\vect{1}_n\times\dot{\vect{q}}_0)) \Big]
\end{align}

When combining $\dot{V} \in \dot{V}_1 +\dot{V}_2 + \dot{V}_3$ notice that $\dot{V}_1$, the term $-\wt{\vect q}^\top(\mat{\Gamma}\mathcal{L}_{22}\otimes\mat{I}_p)\dot{\wt{\vect{q}}}$ of \eqref{eq:V2_dot} and the first summand of \eqref{eq:V3_dot} cancel. Let $\vect x = \wt{\vect q}+\mu\dot{\wt{\vect q}}$ and $\vect y = \vect e + \mu\dot{\vect e}$. Recalling the definition of $\vect s = \wt{\vect q} - \vect e$, we thus have
\begin{align}\label{eq:Vdot_combined}
\dot{V} &\in -\mu^{-1}\mathcal{K}\Big[  \frac{1}{2}\eta\wt{\vect q}^\top \mat X \wt{\vect q} + \frac{1}{2}\mu^2\eta\dot{\wt{\vect q}}^\top \mat X \dot{\wt{\vect q}} - \dot{\wt{\vect q}}^\top \mat{\Gamma}_p \mat M \dot{\wt{\vect q}}\nonumber \\
& \qquad  - \wt{\vect q}^\top \mat{\Gamma}_p \mat C^\top \dot{\wt{\vect q}} +  \vect x^\top \mat{\Gamma}_p \mat C (\vect 1 \otimes \dot{\vect q}_0) -  \vect x^\top \mat{\Gamma}_p \Delta \nonumber \\
& \qquad - \beta \vect x^\top \mat{\Gamma}_p \sgn (\vect x - \vect y)\Big]
\end{align}
From the bounds on $\vect g$, $\mat M$ and $\vect 1 \otimes \ddot{\vect q}_0$, and because we normalised $\bar\gamma = 1$, it follows that \mbox{$\vect x^\top \mat{\Gamma}_p \Delta \leq \xi \Vert \vect x \Vert_2$} where $\xi = k_g + k_{\zeta} + k_{\overline{M}}k_q$.
From Assumption \ref{prty:P4_Cg_bound}, the property of norms and the definition of $\wt{\vect q}$, it follows that $\Vert \mat C \Vert_2 = \Vert \mat C^\top \Vert_2 \leq k_C \Vert \dot{\vect q}\Vert_2 \leq k_C(\Vert \dot{\wt{\vect q}}\Vert_2 + k_p)$. Thus 
\begin{subequations}\label{eq:C_upper_bounds}
\begin{align}
&\wt{\vect q}^\top \mat{\Gamma}_p \mat C^\top \dot{\wt{\vect q}} \leq k_C k_p\Vert \wt{\vect q}\Vert_2 \Vert \dot{\wt{\vect q}}\Vert_2 + k_C \Vert \wt{\vect q}\Vert_2 {\Vert \dot{\wt{\vect q}}\Vert_2}^2 \\
&(\wt{\vect q} + \mu\dot{\wt{\vect q}})^\top \mat{\Gamma}_p \mat C (\vect 1 \otimes \dot{\vect q}_0)  \leq  k_C k_p\Vert \wt{\vect q}\Vert_2 \Vert \dot{\wt{\vect q}}\Vert_2 + \mu k_C k_p{\Vert \dot{\wt{\vect q}}\Vert_2}^2 \nonumber \\
& + \mu k_C k_p^2\Vert \mu^{-1}\wt{\vect q} + \dot{\wt{\vect q}}\Vert_2 
\end{align}
\end{subequations}
Let $\varphi(\mu,\eta) = \frac{1}{2}\mu^2 \eta \lambda_{\min}(\mat{X}) - \mu k_C k_p - k_{\overline{M}}$. Define the functions $\dot{V}_A$ (absolutely continuous) and $\dot{V}_B$ (set-valued) as
\begin{subequations}
\begin{align}
\dot{V}_A & = -\mu^{-1} \big( \varphi(\mu, \eta ){\Vert\dot{\wt{\vect q}}\Vert_2}^2 +\frac{1}{2} \eta \lambda_{\min} (\mat X) {\Vert \wt{\vect q}\Vert_2}^2  \nonumber \\
& \quad -2 k_C k_p\Vert \wt{\vect q}\Vert_2 \Vert \dot{\wt{\vect q}}\Vert_2 - k_C \Vert \wt{\vect q}\Vert_2 {\Vert \dot{\wt{\vect q}}\Vert_2}^2\big)  \label{eq:Vdot_A} \nonumber \\ 
& \quad \triangleq -\mu^{-1} g( \Vert \wt{\vect q}\Vert_2, \Vert\dot{\wt{\vect q}}\Vert_2 ) \\
\dot{V}_B & \in \mu^{-1} \mathcal{K}\Big[\! - \beta\vect x^\top \mat{\Gamma}_p \sgn (\vect x - \vect y) + k_C k_p^2\Vert \vect x\Vert_2 + \xi\Vert \vect x \Vert_2 \Big] \label{eq:Vdot_B}
\end{align}
\end{subequations}
By applying the inequalities in \eqref{eq:C_upper_bounds}, and the eigenvalue inequalities noted in Section \ref{subsec:matrix_theory_background}, we conclude that 
 \begin{align}\label{eq:Vdot_upbound_comb}
 \dot{V} & \leq \dot{V}_A + \dot{V}_B
 \end{align} 

\emph{Part 3:} In \emph{Part 3.1}, we study $\dot{V}_A$ and $\dot{V}_B$ separately to establish negative definiteness properties. Then, \emph{Part 3.2} studies $\dot{V}_A + \dot{V}_B$ and proves a boundedness property. 

\emph{Part 3.1:} Consider the region of the state variables given by $\Vert \wt{\vect q}(t)\Vert_2 \in [0,\mathcal{X}]$ and \mbox{$\Vert \dot{\wt{\vect q}}(t)\Vert_2 \in [0, \infty)$} where $\mathcal{X} > 0$ was computed in Section~\ref{subsec:IC_bound}. One can compute a $\mu^*_5 \geq \mu^*_4$ and $\eta_1^*$ such that $\varphi(\mu,\eta) > 0, \forall \mu \geq \mu^*_5, \eta \geq \eta_1^*$. Note that $\mat L_\mu > \mat H_\mu > \mat N_\mu > 0$ continues to hold. Observe that $g( \Vert \wt{\vect q}\Vert_2, \Vert\dot{\wt{\vect q}}\Vert_2)$ in \eqref{eq:Vdot_A} is of the same form as $g(x,y)$ in Lemma~\ref{lem:bounded_pd_function} with $x = \Vert \wt{\vect q}\Vert_2$ and $y = \Vert \dot{\wt{\vect q}}\Vert_2$. With $b = \varphi(\mu,\eta) > 0$, check if the inequality $\varphi(\mu_5^*, \eta_1^*) > \frac{ (2k_C k_p)^2 }{2 \eta_1^* \lambda_{min}(\mat X)} + k_C \mathcal{X}$ holds. If the inequality holds then $\dot{V}_A$ in \eqref{eq:Vdot_A} is negative definite in the region and proceed to \emph{Part 3.2} However, if the inequality does not hold, then there exists a $\mu_6^* \geq \mu_5^*$ and $\eta_2^* \geq \eta_1^*$ such that
\begin{equation}\label{eq:cond_mu_VdotA_ND}
\varphi(\mu_6^*, \eta_2^*) > \frac{ (2k_C k_p)^2 }{2 \eta_2^* \lambda_{min}(\mat X)} + k_C \mathcal{X}
\end{equation}
Recall from \eqref{eq:V_bar} and \eqref{eq:X_bound} that $\wh{V}^*$ is dependent on $\eta$, but independent of $\mu$ because $\mu_6^* \geq \mu_3^*$. One could leave $\eta_2^* = \eta_1^*$ and find a sufficiently large $\mu_6^*$ to satisfy \eqref{eq:cond_mu_VdotA_ND}. Alternatively, we could increase $\eta$. Notice that $\rho_2(\mu_3^*)$ and $\wh{V}^*$ are both of $\mathcal{O}(\eta)$. Thus, as $\eta$ increases, $\mathcal{X}$ becomes independent of $\eta$, whereas $\varphi = \mathcal{O}(\eta)$. We conclude that there exists a sufficiently large $\eta_2^*$ satisfying \eqref{eq:cond_mu_VdotA_ND}, and for which $\mathcal{X}_1$, $\mathcal{Y}_1$, $\mathcal{X}$, $\mathcal{Y}$ need not be recomputed. With $\mu_6^*, \eta_2^*$ satisfying \eqref{eq:cond_mu_VdotA_ND}, $\dot{V}_A < 0$ in the aforementioned region.

Now consider $\dot{V}_B$ over two time intervals, $ t_\mathcal{P} = [0, T_1)$ and $t_\mathcal{Q} = [T_1, T_2)$, where $T_1$ is given in Theorem~\ref{thm:dist_observer_cao} and $T_2$ is the infimum of those values of $t$ for which one of the inequalities $\Vert \wt{\vect q}(t)\Vert_2 < \mathcal{X}$, $\Vert \dot{\wt{\vect q}}(t)\Vert_2 < \mathcal{Y}$ fails. In \emph{Part 3.2}, we argue that without loss of generality, it is possible to take $T_2> T_1$. In fact, we establish that the inequalities never fail; $T_2$ does not exist and thus $t_{\mathcal{Q}} = [T_1, \infty)$\footnote{Establishing that $T_2$ does not exist rules out the possibility of finite-time escape for system \eqref{eq:euler_lagrange_general_agent}.}. 

Consider firstly $t\in t_\mathcal{P}$. Observe that the set-valued function $-\beta\vect x^\top \mat{\Gamma}_p \sgn(\vect x - \vect y)$ is upper bounded by the single-valued function $\beta \Vert \vect x \Vert_1$. Recalling $\dot{V}_B$ in \eqref{eq:Vdot_B} yields
\begin{align}\label{eq:Vdot_B_tP}
\dot{V}_B & \leq  (\sqrt{n}\beta +k_Ck_p^2+\xi) (\mu^{-1}\Vert \wt{\vect q} \Vert_2 + \Vert \dot{\wt{\vect q}}\Vert_2) := \dot{V}_{\overline{ B}}
\end{align}
because $\Vert \cdot \Vert_2 \leq \Vert \cdot \Vert_1 \leq \sqrt{n}\Vert \cdot \Vert_2$ \cite{horn2012matrixbook}.

For $t\in t_\mathcal{Q}$, Theorem \ref{thm:dist_observer_cao} yields that $\vect e(t)= \dot{\vect e}(t) = \vect 0$, which implies that $\vect y = \vect 0$. Thus, the set-valued term $\mathcal{K}[\vect x^\top \mat{\Gamma}_p \sgn(\vect x-\vect y)]$ in \eqref{eq:Vdot_B} becomes the singleton $\mathcal{K}[\vect x^\top \mat{\Gamma}_p \sgn(\vect x)] = \{\Vert \mat{\Gamma}_p \vect x \Vert _1\}$ (since $\mat{\Gamma}_p > 0$ is diagonal). It then follows that
\begin{align}
\dot{V}_B & = - \mu^{-1} \beta \Vert \mat{\Gamma}_p \vect x \Vert_1 + k_C k_p^2\Vert \vect x\Vert_2 + \xi\Vert \vect x \Vert_2 
\end{align}
In other words, $\dot{V}_B$ for $t\in t_\mathcal{Q}$ is a continuous, single-valued function in the variables $\wt{\vect q}$ and $\dot{\wt{\vect q}}$. For $t\in t_\mathcal{Q}$, we observe that \mbox{$\dot{V}_B \leq - \mu^{-1}(\beta\underline{\gamma} -k_C k_p^2- \xi) \Vert \wt{\vect q} + \mu\dot{\wt{\vect q}} \Vert_1 < 0$} if 
\begin{equation}\label{eq:cond_beta}
\beta > (k_Ck_p^2+\xi)/\underline{\gamma} 
\end{equation}

\emph{Part 3.2:} To aid in this part of the proof, refer to Figure~\ref{fig:stability_diag_tracking}.

Consider firstly $\dot{V}$ for $t\in t_\mathcal{P}$. Specifically, let $\dot{V}_{t_\mathcal{P}} \triangleq \dot{V}_A+\dot{V}_{\overline{B}}$, which gives
\begin{align}\label{eq:Vdot_bound_T1}
\dot{V}_{t_\mathcal{P}} & = -\mu^{-1} \Big[ \varphi(\mu, \eta ){\Vert\dot{\wt{\vect q}}\Vert_2}^2 +\frac{1}{2} \eta \lambda_{\min} (\mat X) {\Vert \wt{\vect q}\Vert_2}^2  \nonumber \\
& \quad -2 k_C k_p\Vert \wt{\vect q}\Vert_2 \Vert \dot{\wt{\vect q}}\Vert_2 - k_C \Vert \wt{\vect q}\Vert_2 {\Vert \dot{\wt{\vect q}}\Vert_2}^2 \nonumber \\ 
& \quad \; - (\sqrt{n} \beta\underline{\gamma} + k_C k_p^2 + \xi) \big(\Vert \wt{\vect q}\Vert_2 + \mu \Vert \dot{\wt{\vect q}} \Vert_2\big) \Big] \nonumber \\ 
& \quad \; \; \triangleq -\mu^{-1} p ( \Vert \wt{\vect q} \Vert_2, \Vert \dot{\wt{\vect q}} \Vert_2 )
\end{align}
Note that $\dot{V} \leq \dot{V}_{t_{\mathcal{P}}}$, i.e. $\dot{V}$ for $t\in t_\mathcal{P}$ is a differential inclusion which is upper bounded by a continuous function. Observe that $p ( \Vert \wt{\vect q} \Vert_2, \Vert \dot{\wt{\vect q}} \Vert_2 )$ is of the form of $h(x,y)$ in Corollary~\ref{cor:bounded_pd_function_2} with $x = \Vert \wt{\vect q}\Vert_2$ and $y = \Vert \dot{\wt{\vect q}}\Vert_2$. Here, $b = \varphi(\mu,\eta)$, $a = \frac{1}{2}\eta\lambda_{\min}(\mat X)$, $c = k_C$, $d = 2 k_C k_p$, $e = (\sqrt{n}\beta +k_Ck_p^2+\xi)$ and $f = \mu e$. Thus, for some given $\vartheta, \varepsilon, \mathcal{X}, \mathcal{Y}$ satisfying the requirements detailed in Corollary~\ref{cor:bounded_pd_function_2}, one can use \eqref{eq:Du_b_ineq} and \eqref{eq:Dv_b_ineq} to find a $\mu,\eta$ such that $p ( \Vert \wt{\vect q} \Vert_2, \Vert \dot{\wt{\vect q}} \Vert_2 )$ is positive definite in the region $\mathcal{R}$. Note that $\vartheta, \varepsilon$ can be selected by the designer. Choose $\vartheta > \mathcal{X} - \mathcal{X}_1$ and $\varepsilon > \mathcal{Y}-\mathcal{Y}_1$, and ensure that $\mathcal{X}-\vartheta, \mathcal{Y}-\varepsilon > 0$. Note the fact that $\mathcal{X}\geq \mathcal{X}_1$ and $\mathcal{Y}\geq \mathcal{Y}_1$ implies $\vartheta, \varepsilon > 0$. 

Define the sets $\mathcal{U}$, $\mathcal{V}$ and the region $\mathcal{R}$ as in Corollary~\ref{cor:bounded_pd_function_2} with $x = \Vert \wt{\vect q}\Vert_2$ and $y = \Vert \dot{\wt{\vect q}}\Vert_2$. Define further sets $\bar{\mathcal{U}} = \{\Vert\wt{\vect q}\Vert_2: \Vert \wt{\vect q}\Vert_2 > \mathcal{X}\}$ and  $\bar{\mathcal{V}} = \{\Vert \dot{\wt{\vect q}}\Vert_2 : \Vert \dot{\wt{\vect q}} \Vert_2 > \mathcal{Y}\}$. Define the compact region $\mathcal{S} = \mathcal{U}\cup \mathcal{V} \setminus \bar{\mathcal{U}}\cup \bar{\mathcal{V}}$, see Fig.~\ref{fig:stability_diag_tracking} for a visualisation of $\mathcal{S}$. Note $\mathcal{S} \subset \mathcal{R}$.
Using Corollary~\ref{cor:bounded_pd_function_2}, and with precise calculation details given in \cite[Theorem 2]{ye2018ELtracking_arXiv}, one can find a pair of gains $\eta$ and $\mu$ which ensures that $p ( \Vert \wt{\vect q} \Vert_2, \Vert \dot{\wt{\vect q}} \Vert_2 )$ is positive definite in $\mathcal{R}$. This implies $p( \Vert \wt{\vect q} \Vert_2, \Vert \dot{\wt{\vect q}} \Vert_2 )$ is positive definite in $\mathcal{S}$. It follows that $\dot{V}_{t_\mathcal{P}}$ is negative definite in $\mathcal{S}$. Further define the region $\Vert \wt{\vect q}(t) \Vert_2 \in [0, \mathcal{X}-\vartheta)$, $\Vert \dot{\wt{\vect q}}(t)\Vert_2\in [0, \mathcal{Y}-\varepsilon)$ as $\mathcal{T}$, again with visualisation in Fig~\ref{fig:stability_diag_tracking}.

Now we justify the fact that we can assume $T_2>T_1$. In fact, in doing so, we show that the existence of $T_2$ creates a contradiction; the trajectories of \eqref{eq:network_system} remain in $\mathcal{T} \cup \mathcal{S}$ for all time. See Fig~\ref{fig:stability_diag_tracking} for a visualisation. Although $\dot{V}$ is sign indefinite in $\mathcal{T}$ (i.e. $V(t)$ can increase), notice from \eqref{eq:V_upper_bound} that, in $\mathcal{T}$ there holds
\begin{align}\label{eq:V_bound_inner}
V(t) & \leq \frac{1}{2}\eta\lambda_{\max}(\mat X) (\mathcal{X}-\vartheta)^2 + \frac{1}{2}k_{\overline{M}}(\mathcal{Y}-\varepsilon)^2 \nonumber \\
&\quad  + \mu^{-1} k_{\overline{M}}(\mathcal{X}-\vartheta)(\mathcal{Y}-\varepsilon) := \mathcal{Z}
\end{align}
Recalling that $\delta > 0$ and is arbitrarily small, one can easily verify that $\mathcal{Z} < \wh{V}^*$ because we selected $\vartheta, \varepsilon$ such that $\mathcal{X}_1 > \mathcal{X} - \vartheta$ and $\mathcal{Y}_1 > \mathcal{Y}- \varepsilon$. In addition, recall $\dot{V}$ is negative definite in $\mathcal{S}$  and now observe the following facts. For any trajectory starting in $\mathcal{S}$ that enters $\mathcal{T}$ at some time $\bar t < T_2$, there holds $V(t) < V(0)$ for all $t \leq \bar t$. Any trajectory starting in $\mathcal{S}$ that stays in $\mathcal{S}$ for all $t$ up to $T_2$ satisfies $V(t) < V(0)$. Any trajectory in $\mathcal{T}$ satisfies $V(t) < \mathcal{Z}$. If any trajectory leaves $\mathcal{T}$ and enters $\mathcal{S}$ at some $\hat t < T_2$, we observe that the crossover point is in the closure of $\mathcal{T}$. Because $V$ is continuous (since the Filippov solutions for $\wt{\vect{q}}, \dot{\wt{\vect q}}$ are absolutely continuous), we have $V(\hat t) \leq \mathcal{Z}$. Because the trajectory enters $\mathcal{S}$, where $\dot{V} < 0$, we also have $V(\hat t + \delta_1) < V(\hat t) \leq \mathcal{Z}$, for some arbitrarily small $\delta_1$. This implies that all trajectories of \eqref{eq:network_system} beginning\footnote{
It is evident from \eqref{eq:initial_X1} that $\wt{\vect q}(0), \dot{\wt{\vect q}}(0) \in \mathcal{S}\cup \mathcal{T}$.} in $\mathcal{T} \cup \mathcal{S}$ satisfy \mbox{$V(t) \leq \max\{\mathcal{Z}, V(0)\} < \wh{V}^*$} for all $t \leq T_2$.


On the other hand, at $T_2$, and in accordance with Lemma~\ref{lem:W_bound}, there holds
\begin{equation}\label{eq:contradict_T2_maximal}
\Vert \wt{\vect q}(T_2) \Vert_2 \leq \sqrt{\frac{V(T_2)}{\chi}} < \sqrt{\frac{\wh{V}^*}{\chi}} < \sqrt{\frac{\wh{V}^*}{\rho_2(\mu^*_3)}} = \mathcal{X}
\end{equation}
where $\chi =\lambda_{\min}(\frac{1}{2}\eta\mat X - \frac{1}{2}\mu^{-2}\mat{\Gamma}_p \mat M ) > \rho_2(\mu^*_3)$.
One can also show that $\Vert \dot{\wt{\vect q}}(T_2)\Vert_2 < \mathcal{Y}$ using an argument paralleling the argument leading to \eqref{eq:contradict_T2_maximal}; we omit this due to spatial limitations. The existence of \eqref{eq:contradict_T2_maximal} and a similar inequality for $\Vert \dot{\wt{\vect q}}(T_2)\Vert_2$ contradicts the definition of $T_2$. In other words, $T_2$ does not exist and $\Vert \wt{\vect q}(t)\Vert_2 < \mathcal{X}$, $\Vert \dot{\wt{\vect q}}(t)\Vert_2 < \mathcal{Y}$ hold for all $t$.

\emph{Part 4:} Observe that $\dot{V}_B$ changes at $t = T_1$ to become negative definite. Consider now $t\in t_\mathcal{Q}=[T_1,T_2)$. Recalling that $\dot{V}\leq \dot{V}_A + \dot{V}_B$, we have
\begin{align}\label{eq:Vdot_tQ_exp}
\dot{V} & \leq -\mu^{-1} \Big[ \varphi(\mu, \eta){\Vert\dot{\wt{\vect q}}\Vert_2}^2 +\frac{1}{2}\eta\lambda_{\min} (\mat X) {\Vert \wt{\vect q}\Vert_2}^2\nonumber \\
& \;  -2 k_C k_p\Vert \wt{\vect q}\Vert_2 \Vert \dot{\wt{\vect q}}\Vert_2 - k_C \Vert \wt{\vect q}\Vert_2 {\Vert \dot{\wt{\vect q}}\Vert_2}^2\big) \nonumber \\
& \quad - (\beta -k_C k_p^2- \xi) \Vert \wt{\vect q}+\mu\dot{\wt{\vect q}}  \Vert_1 \Big]< 0
\end{align} 
in the region $\mathcal{D}: = \mathcal{S}\cup \mathcal{T}$. From the fact that $\Vert \dot{\wt{\vect q}}(T_1)\Vert_2 < \mathcal{Y}$, there holds $\dot{V}(T_1) < 0$. The argument applied to the interval $[0,\min\{T_1,T_2\}]$ above, culminating in \eqref{eq:contradict_T2_maximal}, is now applied to the interval $t_\mathcal{Q}$. Since $\dot{V} < 0$ in $\mathcal{D}$ and at $T_1$, the trajectory is in $\mathcal{D}$, we have $V(T_1) < V(T_2) < \wh{V}^*$. It follows that \eqref{eq:contradict_T2_maximal} continues to hold (and equally for the argument regarding $\Vert\dot{\wt{\vect q}}\Vert_2$). It remains true that $T_2$ does not exist, implying that the trajectory of \eqref{eq:network_system} remains in $\mathcal{D}$ and $\dot{V} < 0$ for $t\in[T_1,\infty)$.  

Recall from below \eqref{eq:cond_V_PD} that the eigenvalues of $\mat{H}_\mu$ are uniformly upper bounded away from infinity and lower bounded away from zero by constants. Specifically, there holds $\lambda_{\min}(\mat{N}_\mu) \Vert {[\wt{\vect q}^\top, \dot{\wt{\vect q}}^\top]^\top\Vert_2}^2 \leq V \leq \lambda_{\max}(\mat{L}_\mu) \Vert {[\wt{\vect q}^\top, \dot{\wt{\vect q}}^\top]^\top\Vert_2}^2$. Because $\mathcal{D}$ is compact, one can find a scalar $a_3 > 0$ such that $\dot{V} \leq - a_3 {\Vert [\wt{\vect q}^\top, \dot{\wt{\vect q}}^\top]^\top\Vert_2}^2 $. It follows that $\dot{V}\leq -[a_3/\lambda_{\max}(\mat{L}_\mu)] V$ in $\mathcal{D}$. This inequality is used to conclude that $V$ decays exponentially fast to zero, with a minimum rate $e^{-a_3/\lambda_{\max}(\mat{L}_\mu) t}$ \cite{khalil2002nonlinear}. Specifically, there holds 
\begin{equation}\label{eq:exp_converg_rate}
\Vert \begin{bmatrix}\wt{\vect q}(t) \\ \dot{\wt{\vect q}}(t)\end{bmatrix}\Vert_2 \leq \frac{\lambda_{\max}(\mat{L}_\mu)}{\lambda_{\min}(\mat{N}_\mu)} \Vert \begin{bmatrix}\wt{\vect q}(0) \\ \dot{\wt{\vect q}}(0)\end{bmatrix}\Vert_2 e^{-\frac{a_3}{\lambda_{\max}(\mat{L}_\mu)}t}
\end{equation}
It follows that $\lim_{t\to\infty} \wt{\vect q}(t) = \vect 0_n$ and $\lim_{t\to\infty} \dot{\wt{\vect q}}(t) = \vect 0_n$ exponentially and the leader tracking objective is achieved.
\end{proof}
\end{theorem}

\begin{figure}[t]
\centering{
\resizebox{0.75\columnwidth}{!}{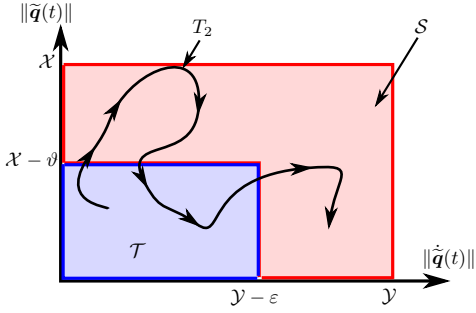}
\caption{Diagram for \emph{Part 3} of the proof of Theorem~\ref{theorem:main_result}. The red region is $\mathcal{S}$, in which $\dot{V}(t) < 0$ for all $t \geq 0$. The blue region is $\mathcal{T}$, in which $\dot{V}(t)$ is sign indefinite. A trajectory of \eqref{eq:network_system} is shown with the black curve. We define $t = T_2$, if it exists, as the infimum of all $t$ values for which one of the inequalities $\Vert \wt{\vect{q}}(T_2) \Vert < \mathcal{X}$ or $\Vert \dot{\wt{\vect q}}(T_2) \Vert < \mathcal{Y}$ fails to hold, i.e. as the time at which the system~\eqref{eq:network_system} first leaves $\mathcal{S}$. By contradiction, it is shown in \emph{Part 3.2} that the trajectory of \eqref{eq:network_system} satisfies $
\Vert \wt{\vect{q}}(T_2) \Vert < \mathcal{X}, \Vert \dot{\wt{\vect q}}(T_2) \Vert < \mathcal{Y}$. I.e., $T_2$ does not exist and the trajectory remains in $\mathcal{T}\cup \mathcal{S}, \forall\,t$. 
The sign indefiniteness of $\dot{V}$ in $\mathcal{T}$ arises due to terms linear in $\Vert \wt{\vect q} \Vert$, $\Vert \dot{\wt{\vect q}} \Vert$ in \eqref{eq:Vdot_bound_T1}. These terms disappear at $t = T_1$, when the finite-time observer converges. For all $t > T_1, \dot{V} < 0$ in $\mathcal{T}\cup \mathcal{S}$ as shown in \emph{Part 4}. Exponential convergence to the origin follows.}
\label{fig:stability_diag_tracking}
}
\end{figure}

\begin{remark}[Additional degree of freedom in gain design]\label{rem:two_design}
In \cite{ye2016ELtracking} we assumed that $\eta = 1$. There is more flexibility in this paper since we allow $\eta > 1$; one can adjust separately, or simultaneously, $\mu$ and $\eta$.   While the interplay between $\mu,\eta,\beta$, and its effect on performance, is difficult to quantify, we observe from extensive simulations that in general, one should make $\beta$ and $\mu$ as small as possible. Where possible, it is better to hold $\mu$ constant and increase $\eta$ to satisfy an inequality involving both, e.g \eqref{eq:cond_mu_VdotA_ND}.  Notice that $\lambda_{\max}(\mat{H}_\mu)$ and $a_3$ are both $\mathcal{O}(\eta)$. Note that as $\mu$ increases $\lambda_{\max}(\mat{H}_\mu)$ does not increase but $a_3$ does decrease. Thus, the convergence rate $a_3/\lambda_{\max}(\mat{L}_\mu)$ is not negatively affected by increasing $\eta$ but is reduced by increasing $\mu$. If only $\mu$ is adjusted to be large (as in \cite{ye2016ELtracking}) then velocity consensus is quickly achieved but position consensus is achieved after a long time. 
\end{remark}

\begin{remark}[Designing the gains]\label{rem:gain_design}
We summarise here the process to design $\mu,\eta,\beta$ to satisfy inequalities detailed in the proof of Theorem~\ref{theorem:main_result}. First, one may select $\beta$ to satisfy \eqref{eq:cond_beta}. Then, $\mu$ should be set to satisfy \eqref{eq:cond_V_PD}. The quantities $\mathcal{X}$ and $\mathcal{Y}$ discussed in Section~\ref{subsec:IC_bound} are then computed with $\eta \gtrsim 1$; we noted below \eqref{eq:cond_mu_VdotA_ND} that $\mathcal{X}$ and $\mathcal{Y}$ are independent of $\eta$ as $\eta$ increases. Having computed $\mathcal{X}$ and $\mathcal{Y}$, the last step is to adjust $\eta$ to ensure $g( \Vert \wt{\vect q}\Vert_2, \Vert\dot{\wt{\vect q}}\Vert_2)$ and $p( \Vert \wt{\vect q}\Vert_2, \Vert\dot{\wt{\vect q}}\Vert_2)$ are positive definite (see Part 3.1 of Theorem~\ref{theorem:main_result} proof). Details of the inequalities to ensure positive definiteness are found in the proofs of Lemma~\ref{lem:bounded_pd_function} and Corollary~\ref{cor:bounded_pd_function_2} in \cite[Section II-A]{ye2018ELtracking_arXiv}.
\end{remark}

\begin{remark}[Robustness]\label{rem:robustness}
The proposed algorithm \eqref{eq:claw_track_sign} is robust in several aspects. First, the exponential stability property implies that small amounts of noise produce small departures from the ideal. Moreover, the signum term in the controller offers robustness to the unknown disturbance $\vect{\zeta}_i(t)$. In contrast, and as discussed in the introduction, adaptive controllers are not robust to unmodelled agent dynamics. 
\end{remark}

\begin{remark}[Controller structure]\label{rem:control_design}
Consider the controller \eqref{eq:claw_track_sign}. The term containing the signum function ensures exact tracking of the leader's trajectory. 
Consider Fig.~\ref{fig:stability_diag_tracking}. For $t < T_1$, the signum term results in the region $\mathcal{T}$, where $\dot{V}$ is sign indefinite. This signum term can in fact drive an agent away from its neighbours due to the nonzero error term $\vect{e}_i(t)$, $t < T_1$. However, for $t < T_1$, the linear terms of the controller (and in particular adjustment of the gains $\eta, \mu$) ensure that $\dot{V} < 0$ in $\mathcal{S}$. This ensures that the followers remain in the bounded region $\mathcal{S}$ centred on the leader.  Such a controller gives added robustness. For example, if $\mathcal{G}_B$ becomes temporarily disconnected, all agents will remain close to the leader so long as $\mathcal{G}_A$ has a directed spanning tree. When connectivity of $\mathcal{G}_B$ is restored, perfect tracking follows; this is illustrated in the simulation below.
\end{remark}

\subsection{Practical Tracking By Approximating the Signum Function}
Although the signum function term in \eqref{eq:claw_track_sign} allows the leader-tracking objective to be achieved, it carries an offsetting disadvantage. Use of the signum function can cause mechanical components to fatigue due to the rapid switching of the control input. Moreover, chattering often results, which can excite the natural frequencies of high-order unmodelled dynamics. A modified controller is now proposed using a continuous approximation of the signum function and we derive an explicit upper bound on the error in tracking of the leader\footnote{We do not consider an approximation for \eqref{eq:dist_observer} because the observer involves \emph{computing} state estimates, as opposed to the \emph{physical} control input for \eqref{eq:euler_lagrange_general_agent}.}. 

Consider the following continuous, model-independent algorithm for the $i^{th}$ agent, replacing \eqref{eq:claw_track_sign}:
\begin{align}\label{eq:claw_track_sig}
\tau_i & = - \eta\sum_{j\in\mathcal{N}_{Ai}} a_{ij}\big( (\vect{q}_i-\vect{q}_j) +\mu(\dot{\vect{q}}_i-\dot{\vect{q}}_j) \big) \nonumber \\
& - \beta \vect z_i \big((\vect q_i- \wh{\vect r}_i)+\mu(\dot{\vect q}_i- \wh{\vect v}_i)\big)
\end{align}
where $\vect z_i(\vect x) \triangleq \vect x/(\Vert \vect x \Vert_2 + \epsilon)$ with $\epsilon > 0$ being the degree of approximation. The function $\vect z_i(\vect x)$ approximates $\sgn(\vect x)$ via the boundary layer concept \cite{edwards1998sliding_book}.
The networked system is:
\begin{align}\label{eq:network_system_sig}
&\mat{M}\ddot{\wt{\vect{q}}} + \mat{C}\dot{\wt{\vect{q}}} + \eta(\mathcal{L}_{22}\otimes\mat{I}_p)(\wt{\vect{q}}+ \mu\dot{\wt{\vect{q}}}) + \vect{g} + \vect{\zeta}\nonumber\\
&  +\beta \vect z (\vect{s}+\mu\dot{\vect s}) + \mat{M}(\vect{1}_n\otimes\ddot{\vect{q}}_0) + \mat{C}(\vect{1}_n\otimes\dot{\vect{q}}_0) = \vect 0
\end{align}
where $\vect z(\vect{s}+\mu\dot{\vect s}) = [\vect z_1(\vect s_1+ \mu \dot{\vect s}_2)^\top, ..., \vect z_n(\vect s_n+ \mu \dot{\vect s}_n)^\top]^\top$. Note that $\Vert \vect z_i (\vect x_i) \Vert_2 < 1$ for any $\epsilon > 0$.

The computation of the quantities $\mathcal{X}, \mathcal{Y}$ in subsection~\ref{subsec:IC_bound} is unchanged. Because of similarity, we do not provide a complete proof here; a sketch is outlined and we leave the minor adjustments to the reader. Consider the same Lyapunov-like function as in \eqref{eq:V_01_quad}, with $\mu$ sufficiently large to ensure $\mat{H}_\mu > 0$. The derivative of \eqref{eq:V_01_quad} with respect to time, along the trajectories of \eqref{eq:network_system_sig}, is given by
\begin{align}\label{eq:Vdot_combined_sig}
\dot{V} & =  -\mu^{-1}\Big[  \frac{1}{2}\eta\wt{\vect q}^\top \mat X \wt{\vect q} + \frac{1}{2}\mu^2\eta\dot{\wt{\vect q}}^\top \mat X \dot{\wt{\vect q}} - \dot{\wt{\vect q}}^\top \mat{\Gamma}_p \mat M \dot{\wt{\vect q}}\nonumber \\
& \qquad  - \wt{\vect q}^\top \mat{\Gamma}_p \mat C^\top \dot{\wt{\vect q}} +  \vect x^\top \mat \Gamma_p \big(\mat C(\vect 1_n \otimes \dot{\vect q}_0 ) + \Delta \big) \nonumber \\ 
& \qquad + \beta \sum_{i = 1}^n \gamma_i \vect x_i^\top \vect z_i(\vect x_i + \vect y_i)\Big]
\end{align}

Let $t_{\mathcal{P}}$ and $t_{\mathcal{Q}}$ be defined as in \emph{Part 3.2} of the proof of Theorem~\ref{theorem:main_result}. One can compute that, for $t \in t_{\mathcal{P}}$, there holds
\begin{align}\label{eq:Vdot_tP_apprx}
\dot{V} & \leq -\mu^{-1} \Big[ \varphi(\mu, \eta ){\Vert\dot{\wt{\vect q}}\Vert_2}^2 +\frac{1}{2} \eta \lambda_{\min} (\mat X) {\Vert \wt{\vect q}\Vert_2}^2  \nonumber \\
& \quad -2 k_C k_p\Vert \wt{\vect q}\Vert_2 \Vert \dot{\wt{\vect q}}\Vert_2 - k_C \Vert \wt{\vect q}\Vert_2 {\Vert \dot{\wt{\vect q}}\Vert_2}^2 - (k_C k_p^2 + \xi)\Vert \vect x\Vert_2   \nonumber \\
& \quad - \beta\sum_{i = 1}^n \Vert \vect x_i \Vert_2 \Big] \leq -\mu^{-1} p( \Vert \wt{\vect q}\Vert_2, \Vert\dot{\wt{\vect q}}\Vert_2 ) 
\end{align}
where $p(\cdot, \cdot)$ was defined in \eqref{eq:Vdot_bound_T1}. This is because there holds $\Vert \vect x_i^\top \vect z_i(\vect x_i + \vect y_i) \Vert_2 < \Vert \vect x_i \Vert_2$, and $\bar\gamma = 1$. In other words, any $\mu, \eta$ which ensures boundedness of the trajectories of \eqref{eq:network_system}, will also ensure that the trajectories of \eqref{eq:network_system_sig} remain bounded in $\mathcal{S} \cup \mathcal{T}$ for all time.

Consider now $t\in t_{\mathcal{Q}}$, and observe that $\vect x_i^\top \vect z_i(\vect x_i) = {\Vert \vect x_i \Vert_2}^2/(\Vert \vect x_i \Vert_2 + \epsilon)$. It follows that
\begin{align}
\vect x^\top &  \mat \Gamma_p \big(\mat C(\vect 1_n \otimes \dot{\vect q}_0 ) + \Delta \big) + \beta \sum_{i = 1}^n \gamma_i \vect x_i^\top \vect z_i(\vect x_i) \nonumber \\
& = \sum_{i = 1}^n \gamma_i \vect x_i^\top (\mat C_i \dot{\vect q}_0 + \mat M_i \ddot{\vect q}_0 + \vect g_i + \vect \zeta_i ) + \beta \gamma_i \vect x_i^\top \vect z_i(\vect x_i) \nonumber \\
& \leq  \sum_{i = 1}^n \gamma_i \Big[ - (k_C {k_p}^2 + \xi) \Vert \vect x_i \Vert_2 + \beta\frac{ {\Vert \vect x_i \Vert_2}^2 }{\Vert \vect x_i \Vert_2 + \epsilon}  \Big] \nonumber \\
& \quad + k_C k_p\Vert \wt{\vect q}\Vert_2 \Vert \dot{\wt{\vect q}}\Vert_2 + \mu k_C k_p{\Vert \dot{\wt{\vect q}}\Vert_2}^2
\end{align}
which in turn yields
\begin{align}\label{eq:Vdot_tQ_apprx}
\dot{V} & \leq -\mu^{-1} \Big[ \varphi(\mu, \eta ){\Vert\dot{\wt{\vect q}}\Vert_2}^2 +\frac{1}{2} \eta \lambda_{\min} (\mat X) {\Vert \wt{\vect q}\Vert_2}^2  \nonumber \\
& \quad -2 k_C k_p\Vert \wt{\vect q}\Vert_2 \Vert \dot{\wt{\vect q}}\Vert_2 - k_C \Vert \wt{\vect q}\Vert_2 {\Vert \dot{\wt{\vect q}}\Vert_2}^2   \nonumber \\
& \quad - \sum_{i = 1}^n (k_C k_p^2 + \xi)\Vert \vect x_i \Vert_2  + \beta\underline{\gamma} \sum_{i = 1}^n \frac{ {\Vert \vect x_i \Vert_2}^2 }{\Vert \vect x_i \Vert_2 + \epsilon} \Big] 
\end{align}
If $\beta$ satisfies \eqref{eq:cond_beta} then there holds
\begin{align}
\beta\underline{\gamma} & \sum_{i = 1}^n \frac{ {\Vert \vect x \Vert_2}^2 }{\Vert \vect x_i \Vert_2 + \epsilon}- \sum_{i = 1}^n (k_C k_p^2 + \xi)\Vert \vect x_i \Vert_2 \nonumber \\
& \geq \beta \underline{\gamma} \sum_{i = 1}^n \Big[ \frac{ {\Vert \vect x_i \Vert_2}^2 }{\Vert \vect x_i \Vert_2 + \epsilon} - \Vert \vect x_i \Vert_2 \Big]  \\
& = -\beta \underline{\gamma} \sum_{i = 1}^n \Big[ \frac{ {\Vert \vect x_i \Vert_2} \epsilon }{\Vert \vect x_i \Vert_2 + \epsilon} \Big] > -\beta \underline{\gamma} n \epsilon 
\end{align}
because $\Vert \vect x_i \Vert_2/(\Vert \vect x_i \Vert_2 + \epsilon) < 1$ for all $\epsilon > 0$. From this, we conclude that $\dot{V} \leq \dot{V}_A + \beta\underline{\gamma}n \epsilon$. Recall also that any $\mu, \eta$ which ensures $p(\cdot, \cdot)$ is positive definite in $\mathcal{S}$ also ensures that $\dot{V}_A$ is negative definite in $\mathcal{D}$. Similar to \emph{Part 4} of the proof of Theorem~\ref{theorem:main_result}, one has $\dot{V} \leq \psi V + \beta \underline{\gamma} n \epsilon$, for some $\psi > 0$. We conclude using  \cite[Lemma 3.4 (Comparison Lemma)]{khalil2002nonlinear} that
\begin{align}
V(t) & \leq V(0) e^{-\psi t} + \beta\underline{\gamma} n \epsilon \int^t_n e^{-\psi(t-\tau)} . \text{d}\tau \\
& \leq e^{-\psi t} \big[ V(0) + \beta\underline{\gamma} n \epsilon/\psi \big] + \beta\underline{\gamma} n \epsilon/\psi
\end{align}
which implies that $V(t)$ decays exponentially fast to the bounded set $\{ [\wt{\vect q}^\top, \dot{\wt{\vect q}}^\top]^\top : V \leq \beta\underline{\gamma} n \epsilon/\psi \}$. From the fact that $V \geq \lambda_{\min}(\mat N_\mu) \Vert {[\wt{\vect q}^\top, \dot{\wt{\vect q}}^\top]^\top \Vert_2}^2$, we conclude that the trajectories of \eqref{eq:network_system_sig} converge to the bounded set
\begin{equation}
\Omega = \left\{ [\wt{\vect q}^\top, \dot{\wt{\vect q}}^\top] : \Vert [\wt{\vect q}^\top, \dot{\wt{\vect q}}^\top] \Vert_2 \leq \Big(\frac{\beta \underline{\gamma} n \epsilon}{\psi \lambda_{\min}(\mat N_\mu) }\Big)^{ \frac{1}{2} } \right\}
\end{equation}


\section{Leader Tracking on Dynamic Networks}\label{sec:switch_top}
In this section, we consider the case where \emph{the sensing graph $\mathcal{G}_A(t)$ is dynamic, i.e. time-varying.} We assume that there is a finite set $\mathcal{J}$ of $m$ possible network topologies, given as $\bar{\mathcal{G}}_{A} = \{\mathcal{G}_{A,j} = (\mathcal{V}, \mathcal{E}_j, \mathcal{A}_j) : j \in \mathcal{J} \}$, where $\mathcal{J} = \{1, \hdots, m\}$ is the index set. We assume further that $\mathcal{G}_{A,j}$, $\forall\,j$, contains a directed spanning tree, with $v_0$ as the root node and with no edges incoming to $v_0$. Define $\sigma(t) : [0, \infty) \mapsto  \mathcal{J}$ as the piecewise constant switching signal which determines the switching of $\mathcal{G}_A(t)$, with a finite number of switches. The switching times are indexed as $t_1, t_2, \hdots$ and we assume that $\sigma(t)$ is such that $t_{i+1} - t_i > \pi_d > 0$ for all $i$, where $\pi_d$ is the dwell time.  

The dynamic network is modelled by the graph $\mathcal{G}_A(t) = \mathcal{G}_{A,\sigma(t)}$, which in turn implies that the Laplacian associated with $\mathcal{G}_{A}(t)$ is dynamic, given by $\mathcal{L}_A(t) = \mathcal{L}_{A,\sigma(t)}$. Denote $\mathcal{L}_{22}(t) = \mathcal{L}_{22,\sigma(t)}$ as the lower block matrix of $\mathcal{L}_A(t)$, partitioned as in \eqref{eq:laplacian_partition}. It is straightforward to show that the follower network dynamics is given by
\begin{align}\label{eq:network_system_dyn_top}
&\ddot{\wt{\vect{q}}} \in ^{a.e.} \mathcal{K} \Big[ -\mat{M}^{-1}\big[ \mat{C}\dot{\wt{\vect{q}}} + \eta(\mathcal{L}_{22,\sigma(t)}\otimes\mat{I}_p)(\wt{\vect{q}}+ \mu\dot{\wt{\vect{q}}}) + \vect{g} + \vect{\zeta}\nonumber\\
&  +\beta\sgn\left(\vect s +\mu\dot{\vect s}\right) + \mat{M}(\vect{1}_n\otimes\ddot{\vect{q}}_0) + \mat{C}(\vect{1}_n\otimes\dot{\vect{q}}_0)\big]\Big]
\end{align} 

We now seek to exploit an established result which states that a switched system is exponentially stable if the switching is sufficiently slow \cite{liberzon2012switching}, and its `frozen' versions of the various systems arising between switching instants are all exponentially stable. Specifically, the following result holds
\begin{theorem}[{\cite[Theorem 3.2]{liberzon2012switching}}]\label{thm:switch_stable}
Consider the family of systems $\dot{\vect{x}} = \vect{f}_j(\vect{x}), j\in\mathcal{J}$. Suppose that, in a domain $D \subseteq \mathbb{R}^n$ containing $\vect{x} = \vect{0}$, $\exists \mathcal{C}^1$ functions $V_j : D \mapsto \mathbb{R}, j\in \mathcal{J}$, and positive constants $c_j, d_j$, $\Lambda_j$ such that
\begin{equation}
c_j {\Vert \vect{x}\Vert_2}^2 \leq V_j(\vect{x}) \leq d_j {\Vert \vect{x}\Vert_2}^2\, , \quad \forall\,x\in D,\; \forall\,j \in \mathcal{J}
\end{equation}
and $\dot{V}_j(\vect{x}) \leq - \Lambda_j V_j(\vect{x})\, , \quad \forall\,x\in D,\; \forall\, j \in \mathcal{J} $. Define $\kappa \triangleq \sup_{p,q \in \mathcal{J}} \{ V_p(\vect{x})/V_q(\vect{x}) : \vect{x} \in D \}$, and suppose further that $0< \kappa<1$. Then, for $\vect{x}(0) \in D$, the origin $\vect{x} = \vect{0}$ of the switched system $\dot{\vect{x}} = \vect{f}_{\sigma(t)}(\vect{x})$ is exponentially stable for every switching signal $\sigma(t)$ with dwell time $\pi_d > \log(\kappa)/\Lambda$, where $\Lambda = \min_{j\in\mathcal{J}} \Lambda_j$. 
\end{theorem}

Under Assumptions~\ref{assm:leader_traj} and \ref{assm:IC}, we know from the previous Theorem~\ref{theorem:main_result} that for each $j^{th}$ subsystem, 
\begin{align}\label{eq:network_system_dyn_top_j}
&\ddot{\wt{\vect{q}}} \in ^{a.e.} \mathcal{K} \Big[ -\mat{M}^{-1}\big[ \mat{C}\dot{\wt{\vect{q}}} + \eta(\mathcal{L}_{22,j}\otimes\mat{I}_p)(\wt{\vect{q}}+ \mu\dot{\wt{\vect{q}}}) + \vect{g} + \vect{\zeta}\nonumber\\
&  +\beta\sgn\left(\vect s +\mu\dot{\vect s}\right) + \mat{M}(\vect{1}_n\otimes\ddot{\vect{q}}_0) + \mat{C}(\vect{1}_n\otimes\dot{\vect{q}}_0)\big]\Big]
\end{align} 
there exist control gains $\mu_j, \eta_j, \beta_j$ which exponentially achieve the leader tracking objective. In seeking to apply Theorem~\ref{thm:switch_stable} to the system \eqref{eq:network_system_dyn_top}, we obtain, for each $j\in\mathcal{J}$ with $V_j$ given in \eqref{eq:V_01_quad}, the values $\lambda_{\min}(\mat{N}_{\mu,j}) = c_j$, $\lambda_{\max}(\mat{L}_{\mu,j}) = d_j$ and $\Lambda_j = a_{3,j}/\lambda_{\max}(\mat{L}_{\mu,j})$ where $a_{3,j}$ was computed below \eqref{eq:Vdot_tQ_exp}. It follows that $\Lambda = \min_{j\in\mathcal{J}} a_{3,j}/\lambda_{\max}(\mat{L}_{\mu,j})$, and one can also obtain that $\kappa = \max_{j\in\mathcal{J}} \lambda_{\max}(\mat{L}_{\mu,j})/ \min_{j\in\mathcal{J}} \lambda_{\min}(\mat{N}_{\mu,j})$.

\begin{theorem}
Under Assumptions~\ref{assm:leader_traj} and \ref{assm:IC}, with dynamic topology given by $\mathcal{G}_A(t) = \mathcal{G}_{A,\sigma(t)}$, the leader tracking objective is achieved using \eqref{eq:claw_track_sign} if $1)$ the control gains $\mu, \eta, \beta$ satisfy a set of lower bounding inequalities, and $2)$ the dwell time $\pi_d$ satisfies the inequality $\pi_d > \log(\kappa)/\Lambda$, where $\kappa, \Lambda$ are as defined in the immediately preceding paragraph. 
\end{theorem}
\begin{proof}
By selecting $\mu = \max_{j\in \mathcal{J}} \mu_j$, $\eta = \max_{j\in \mathcal{J}} \eta_j$, and $\beta = \max_{j\in \mathcal{J}} \beta_j$, we guarantee each $j^{th}$ subsystem \eqref{eq:network_system_dyn_top_j} is exponentially stable, and also guarantee the boundedness of the trajectories of \eqref{eq:network_system_dyn_top} before the finite-time observer has converged. After convergence of the finite-time observer, application of Theorem~\ref{thm:switch_stable} using the quantities of $\kappa$ and $\Lambda$ outlined above delivers the conclusion that \eqref{eq:network_system_dyn_top} is exponentially stable, i.e. the leader tracking objective is achieved.
\end{proof}

\section{Simulations}\label{sec:sim}

A simulation is now provided to demonstrate the algorithm \eqref{eq:claw_track_sign}. Each agent is a two-link robotic arm and five follower agents must track the trajectory the leader agent. The equations of motion are given in \cite[pp. 259-262]{spong2006robot}. The generalised coordinates for agent $i$ are $\vect{q}_i = [q_i^{(1)}, q_i^{(2)}]^\top$, which are the angles of each link in radians. The agent parameters and initial conditions are given in Table~\ref{tab:agent_param}, and are chosen arbitrarily. Several aspects of the simulation are designed to highlight the robustness of the algorithm. First, the topology is assumed to be switching, with the graph $\mathcal{G}_A(t)$ switching periodically between the three graphs indicated in Fig.~\ref{fig:GA_switch}, at a frequency of $1\;\mathrm{Hz}$. Graph $\mathcal{G}_B(t)$ switches between the three graphs indicated in Fig.~\ref{fig:GB_switch}, also at a frequency of $1\;\mathrm{Hz}$. Moreover, if $\mathcal{G}_{A}(t) = \mathcal{G}_{A,i}$ then  $\mathcal{G}_{B}(t) = \mathcal{G}_{B,i}$ for $i = 1, 2, 3$. Additionally, $\mathcal{G}_B(t)$ is entirely disconnected for $t\in [10,20)$ of the simulation. Last, each agent has a disturbance $\vect\zeta_i(t) = [\sin(i\times0.1t), \cos(i\times0.1t)]^\top$ for $i = 1,\hdots, 5$.  All edges of $\mathcal{G}_A(t)$ and $\mathcal{G}_B(t)$ have edge weights of $5$. The control gains are set as $\mu = 1.5$, $\eta = 16$, $\beta = 25$; they are first computed using the inequalities and then adjusted because the inequalities can lead to conservative gain choices. For the observer, set $\omega_1 = 1$, $ \omega_2 = 5$. The leader trajectory is
\begin{equation*}
\vect q_0(t) =\begin{bmatrix}
0.5\sin(t) - 0.2\sin(0.5t) \\ 0.4\left(2\sin(t) + \frac{\sin(2t)}{(2)} +\frac{\sin(3t)}{(3)}+ \frac{\sin(4t)}{(4)}\right)
\end{bmatrix}
\end{equation*} 
Figure \ref{fig:trajtrack_q} shows the generalised coordinates $q^{(1)}$ and $q^{(2)}$. The generalised velocities, $\dot{q}^{(1)}$ and $\dot{q}^{(2)}$ are shown in Fig. \ref{fig:trajtrack_qdot}. The well studied observer results are omitted \cite{cao2010SL_estimators}. Consider Fig.~\ref{fig:trajtrack_q}. Clearly, $\vect{q}_i(t)$ has almost tracked the leader by $t=10$, but the distributed observer graph $\mathcal{G}_B(t)$ disconnects for $t\in[10,20)$. As discussed in Remark~\ref{rem:control_design}, the controller \eqref{eq:claw_track_sign} has robustness to network failure, since the linear term in \eqref{eq:claw_track_sign} ensures the trajectories remain bounded as long as $\mathcal{G}_B(t)$ is disconnected (thus followers do not possess accurate knowledge of $\vect{q}_0, \dot{\vect{q}}_0$). In the simulation, we observe leader tracking is achieved once $\mathcal{G}_B(t)$ reconnects at $t>20$. Figures~\ref{fig:trajtrack_q_high_mu} and \ref{fig:trajtrack_qdot_high_mu} show the generalised coordinates and generalised velocity, respectively, for the same simulation set up but with an increase of $\mu = 4$ from $\mu = 1.5$. The effects are clear, when we compare to Fig.~\ref{fig:trajtrack_q} and \ref{fig:trajtrack_qdot}. First, the rate of velocity synchronisation relative to the rate of position synchronisation is much larger when $\mu=4$. On the other hand, overall convergence rate is decreased; it takes longer for position and velocity synchronisation to occur, with reasons presented in Remark~\ref{rem:two_design}. However, the increased $\mu$ has a benefit of making the follower agents stay in a smaller ball around the leader when $\mathcal{G}_B(t)$ is disconnected for $t\in[10,20)$, i.e. the tracking error for $t\in[10,20)$ is smaller. This is because increasing $\mu$ decreases the size of $\mathcal{T}$, where $\dot{V}$ is sign indefinite, as shown in Fig.~\ref{fig:stability_diag_tracking}. Last, we show a simulation which utilises the continuous approximation algorithm \eqref{eq:claw_track_sig}. The simulation setup is given above, and we let $\epsilon = 0.5$. Figure~\ref{fig:trajtrack_q_sig_magnify} shows the generalised coordinates of the resulting simulation, with a magnification of the plot for the final 2 seconds of simulation. One can clearly see that practical tracking is achieved, with a small error. The velocity plot is omitted.

\begin{table*}
\caption{Agent parameters used in simulation}
\label{tab:agent_param}
\begin{center}
\vspace*{-11pt}
{\renewcommand{\arraystretch}{1.1}
\begin{tabular}{c c c c c c c c c c c c c}
\scriptsize • & $m_1$ & $m_2$ & $l_1$ & $l_2$ & $l_{c1}$ & $l_{c2}$ & $I_1$ & $I_2$ & $q_i^{(1)}(0)$ & $q_i^{(2)}(0)$ & $\dot{q}_i^{(1)}(0)$ & $\dot{q}_i^{(2)}(0)$  \\ 
\hline 
Agent 1 & 0.5 & 0.4 & 0.4 & 0.3 & 0.2 & 0.15 & 0.1 & 0.05 & 0.1 & 0.9 & -0.5 & -0.6 \\ 
\hline 
Agent 2 & 0.2 & 0.4 & 0.6 & 0.1 & 0.35 & 0.08 & 0.15 & 0.08 & -0.4 & 0.9 & 0.1 & -1.4\\ 
\hline 
Agent 3 & 0.5 & 0.4 & 0.4 & 0.3 & 0.2 & 0.15 & 0.1 & 0.05 & 0.9 & -1.2 & 0.3 & 0.6\\ 
\hline 
Agent 4 & 1 & 0.6 & 0.45 & 0.8 & 0.2 & 0.4 & 0.15 & 0.5 & -2.0 & -2.0 & -1.0 & 0.2\\ 
\hline 
Agent 5 & 0.25 & 0.4 & 0.8 & 0.5 & 0.3 & 0.1 & 0.45 & 0.15 & 0.3 & 1.5 & 1.0 & 1.2\\ 
\hline 
\end{tabular}
}
\end{center}
\end{table*} 

\begin{figure}[!t]
\begin{center}
\includegraphics[width=0.6\linewidth]{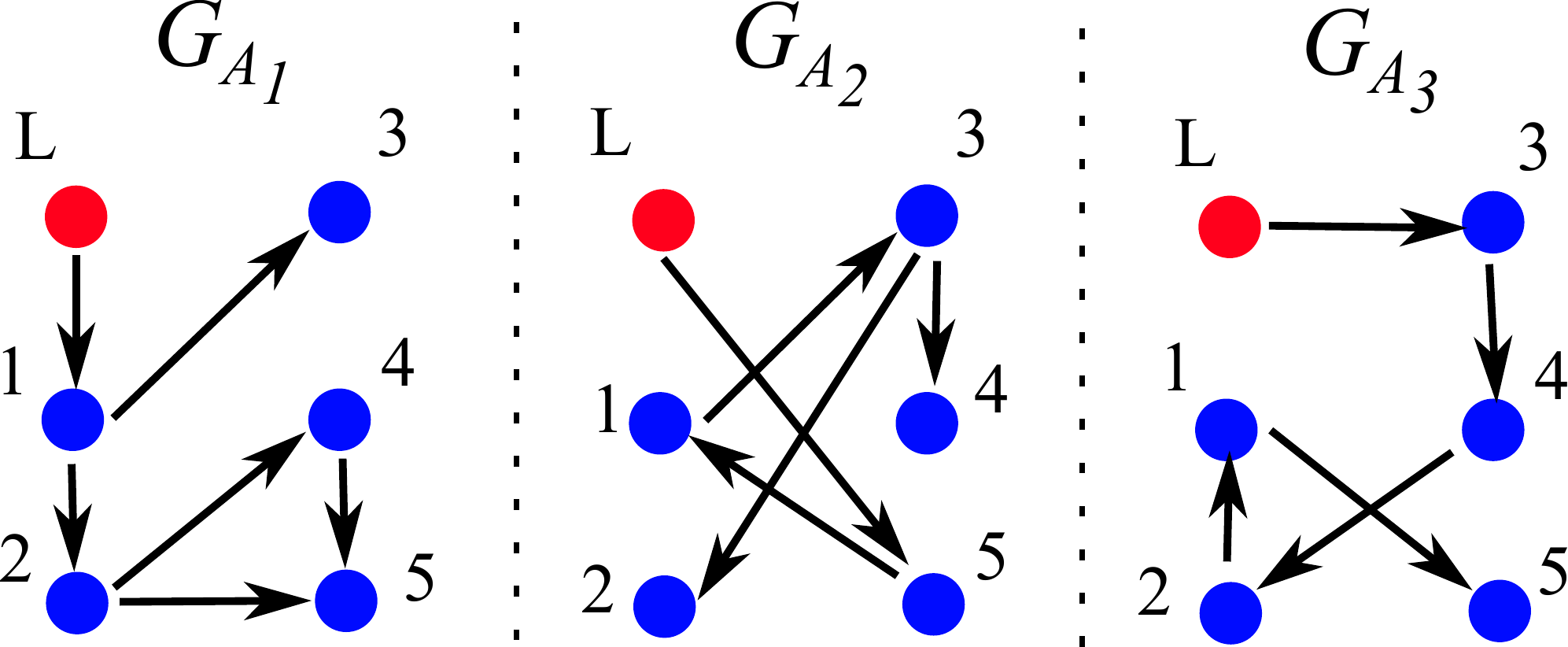}
\caption{In the simulation, graph $\mathcal{G}_{A}(t)$ switches between the above three graphs periodically at a rate of $1\;\textrm{Hz}$.}
\label{fig:GA_switch}
\end{center}
\end{figure}

\begin{figure}[!t]
\begin{center}
\includegraphics[width=0.6\linewidth]{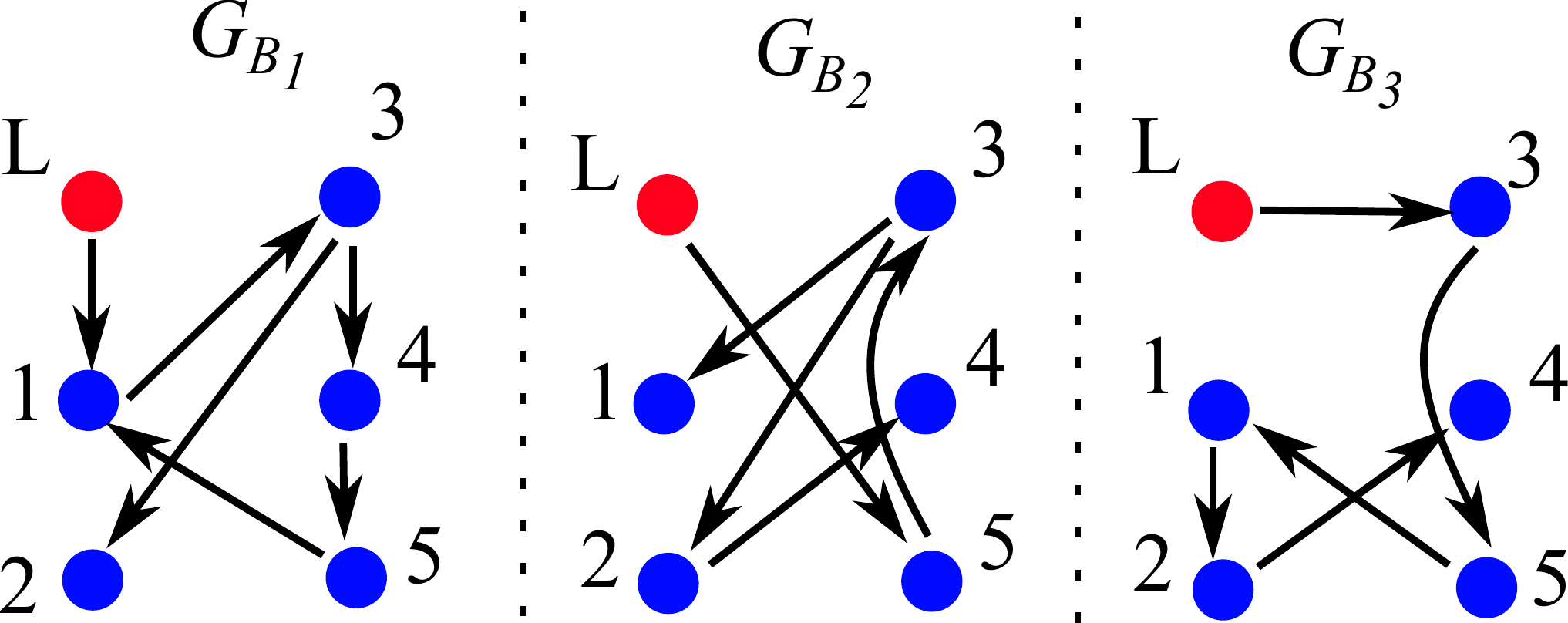}
\caption{Graph $\mathcal{G}_{B}(t)$ switches between the above three graphs periodically at a rate of $1\;\textrm{Hz}$; if $\mathcal{G}_{A}(t) = \mathcal{G}_{A,i}$ then  $\mathcal{G}_{B}(t) = \mathcal{G}_{B,i}$ for $i = 1, 2, 3$.}
\label{fig:GB_switch}
\end{center}
\end{figure}

\begin{figure}
\begin{center}
\vspace*{-11pt}
\includegraphics[height=\columnwidth, angle=-90]{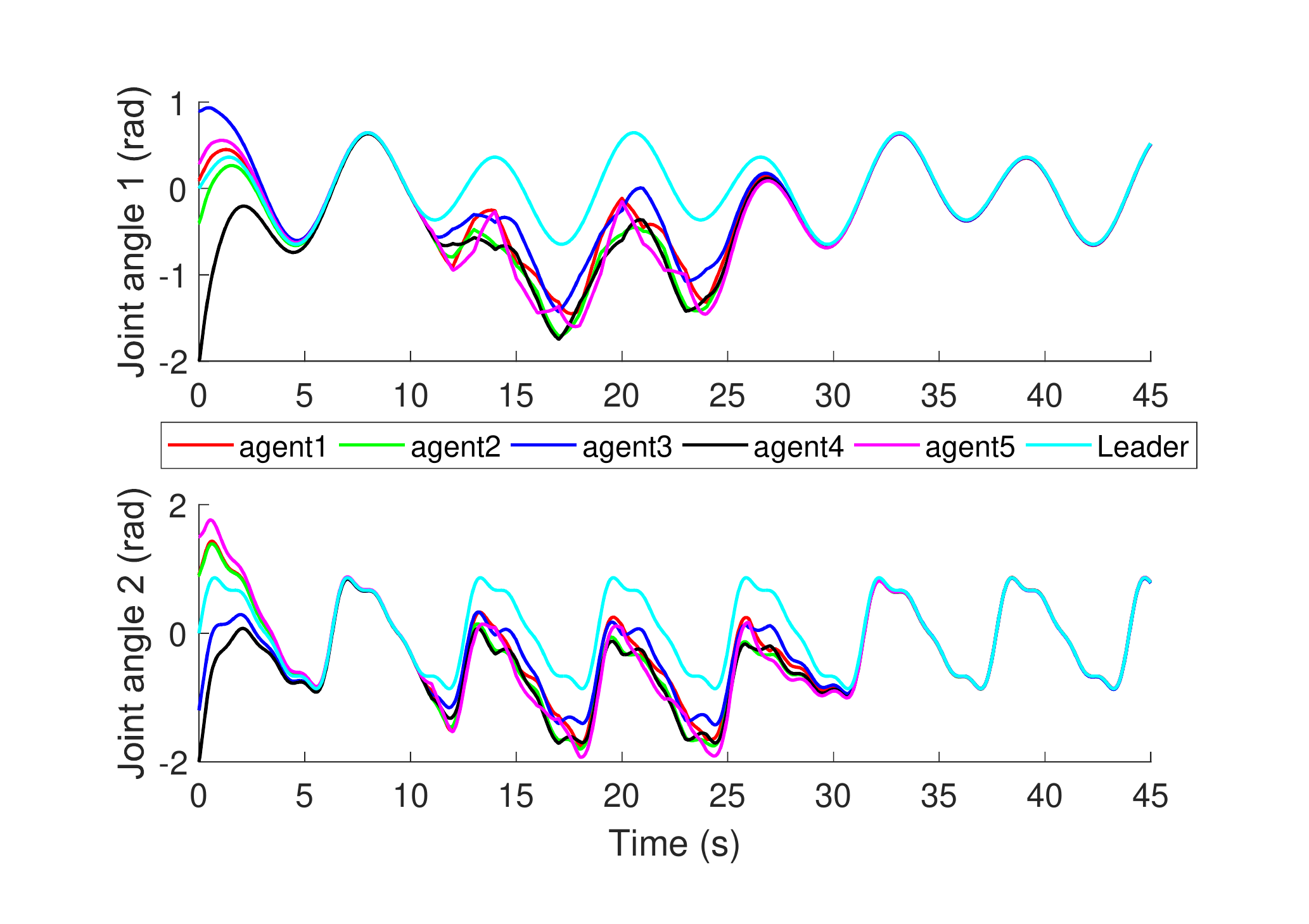}
\caption{Plot of generalised coordinates vs. time; the graph $\mathcal{G}_B(t)$ is disconnected for $t\in[10,20)$.}
\label{fig:trajtrack_q}
\vspace*{-11pt}
\end{center}
\end{figure}
\begin{figure}
\begin{center}
\vspace*{-11pt}
\includegraphics[height=\columnwidth, angle=-90]{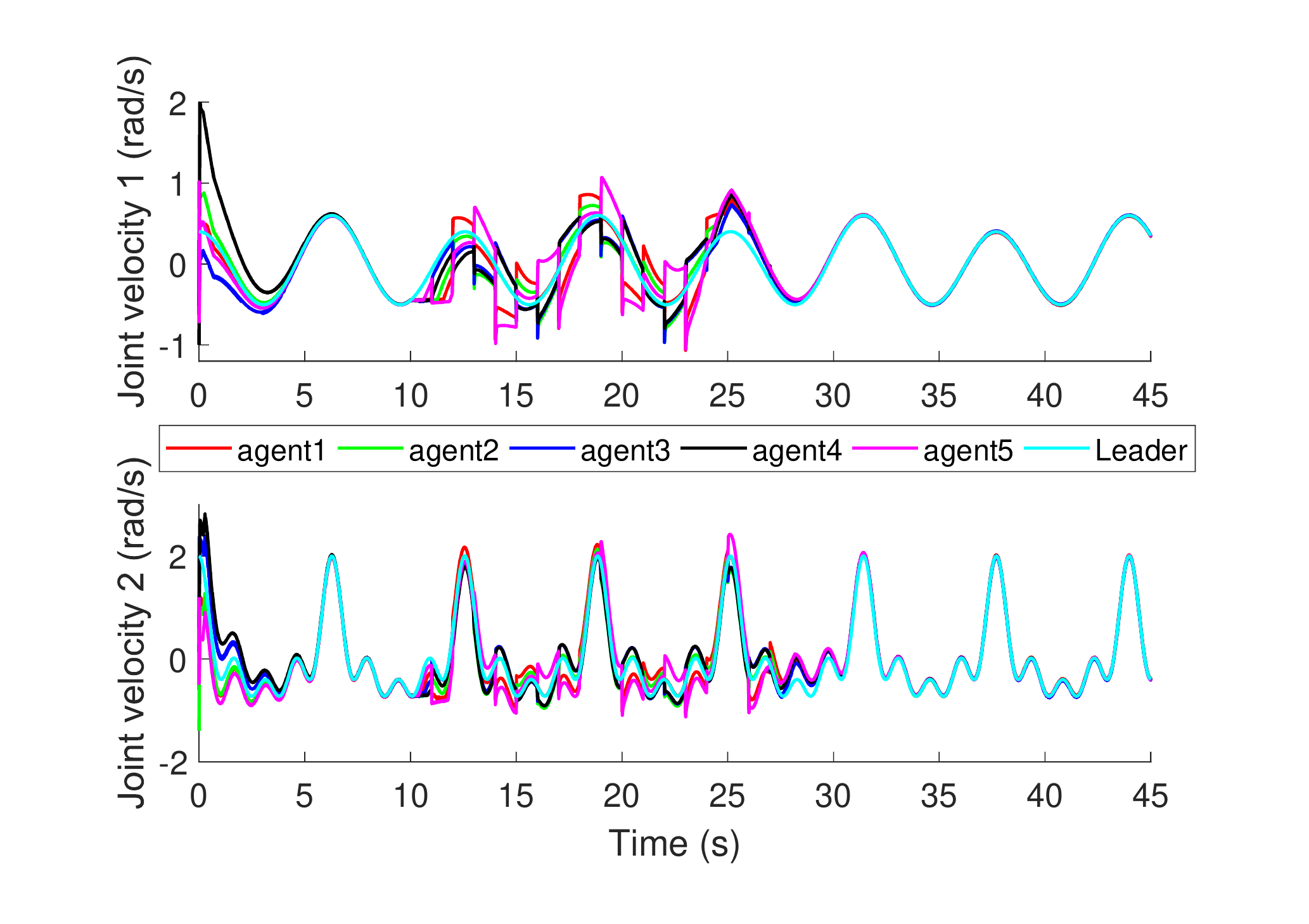}
\caption{Plot of generalised velocity vs. time; the graph $\mathcal{G}_B(t)$ is disconnected for $t\in[10,20)$.}
\label{fig:trajtrack_qdot}
\vspace*{-11pt}
\end{center}
\end{figure}

\begin{figure}
\begin{center}
\vspace*{-11pt}
\includegraphics[height=\columnwidth, angle=-90]{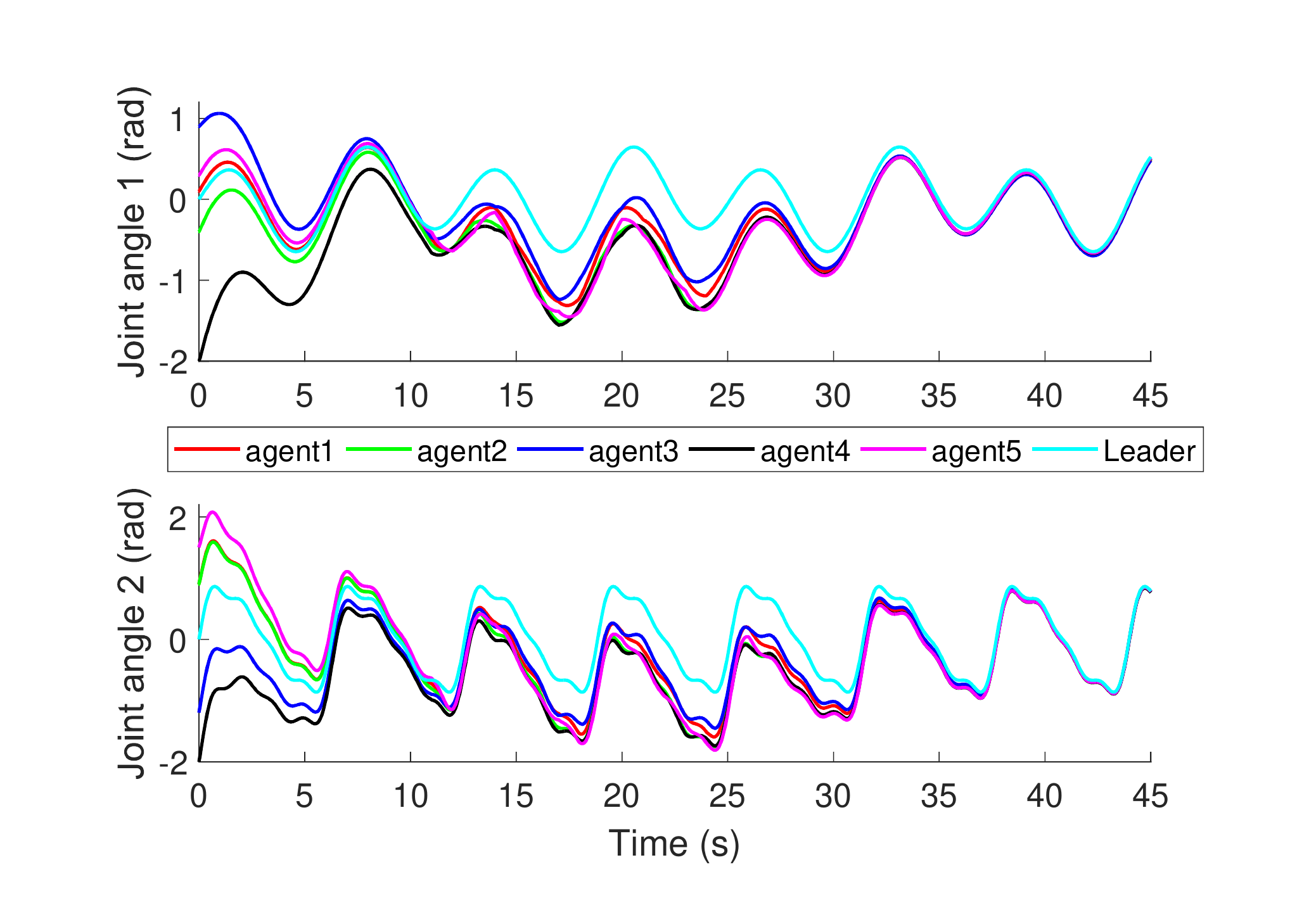}
\caption{Plot of generalised coordinates vs. time; the graph $\mathcal{G}_B(t)$ is disconnected for $t\in[10,20)$. The gain $\mu$ has been increased from $\mu=1.5$ to $\mu=4$.}
\label{fig:trajtrack_q_high_mu}
\vspace*{-11pt}
\end{center}
\end{figure}

\begin{figure}
\begin{center}
\vspace*{-11pt}
\includegraphics[height=\columnwidth, angle=-90]{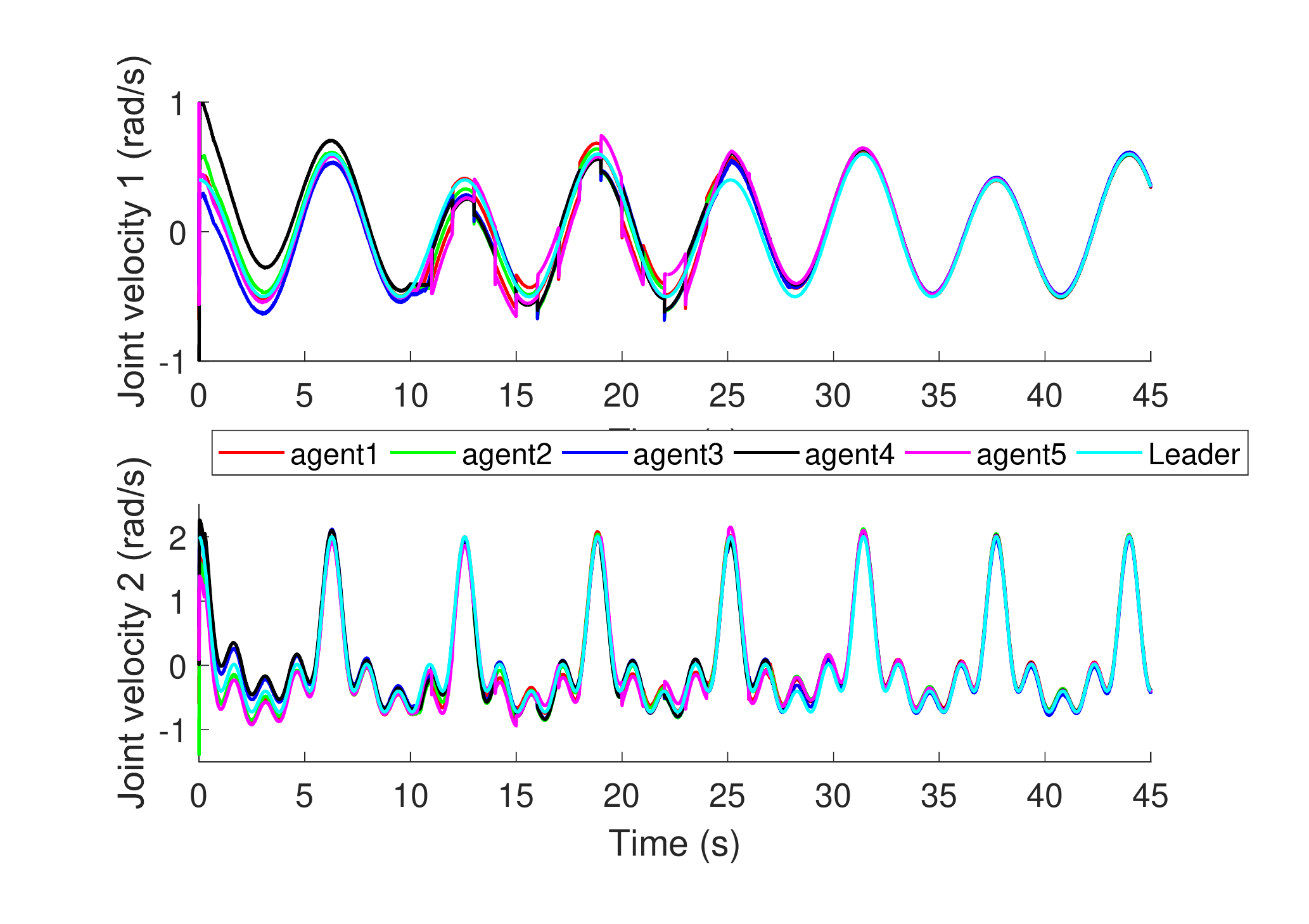}
\caption{Plot of generalised velocity vs. time; the graph $\mathcal{G}_B(t)$ is disconnected for $t\in[10,20)$. The gain $\mu$ has been increased from $\mu=1.5$ to $\mu=4$.}
\label{fig:trajtrack_qdot_high_mu}
\vspace*{-11pt}
\end{center}
\end{figure}

\begin{figure}
\begin{center}
\vspace*{-11pt}
\includegraphics[height=\columnwidth, angle=-90]{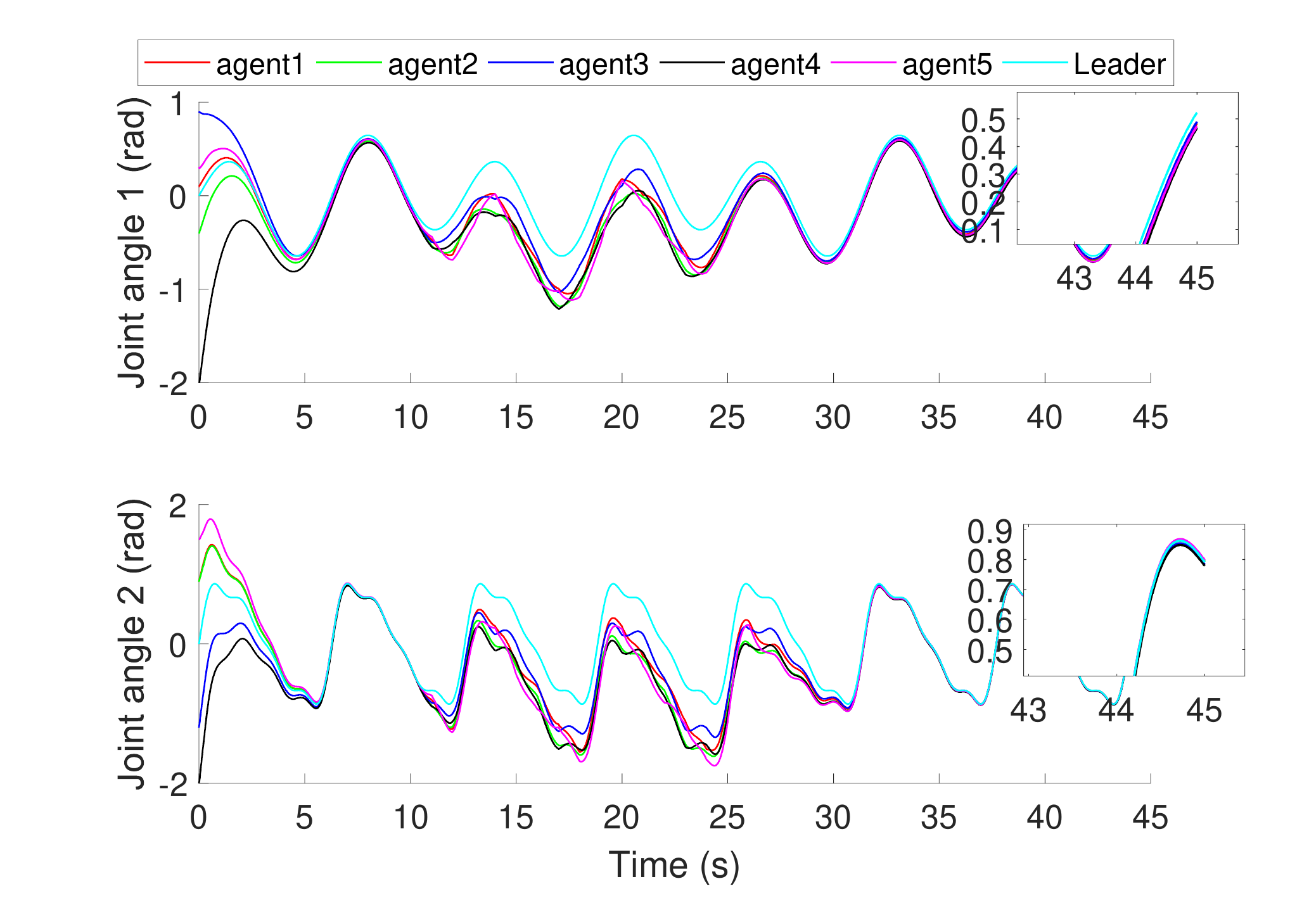}
\caption{Plot of generalised coordinates vs. time; the graph $\mathcal{G}_B(t)$ is disconnected for $t\in[10,20)$. The continuous approximation algorithm \eqref{eq:claw_track_sig} is used with $\epsilon = 0.5$.}
\label{fig:trajtrack_q_sig_magnify}
\vspace*{-11pt}
\end{center}
\end{figure}

\section{Conclusion}\label{sec:conclusion}
In this paper, a distributed, discontinuous model-independent algorithm was proposed for a directed network of Euler-Lagrange agents. It was shown that the leader tracking objective is achieved semi-globally exponentially fast if the directed graph contains a directed spanning tree, rooted at the leader, and if three control gains satisfied a set of lower bounding inequalities. The algorithm was shown to be robust to agent disturbances, unmodelled agent dynamics and modelling uncertainties. A continuous approximation of the algorithm was proposed to avoid chattering, and we then extended the result to include switching topologies. A numerical simulation illustrated the algorithm's effectiveness.

\ifCLASSOPTIONcaptionsoff
  \newpage
\fi



%
%
%

\bibliographystyle{IEEEtran}
\bibliography{MYE_ANU}
%

%
%
%




\end{document}